\documentclass[12pt,a4paper]{article}
\pdfoutput=1
 \usepackage{jheppub}
 \usepackage{amsmath,amsthm}
 \usepackage{amsfonts}
 \usepackage{mathtools}
 \usepackage{amssymb}
 \usepackage{caption}
 \usepackage{subcaption}
 \usepackage{slashed}
\usepackage{float}
\usepackage{graphicx}
\graphicspath{ {project2 tex/} }

\newcommand{\be} {\begin{equation}}
\newcommand{\ee} {\end{equation}}
\newcommand{\bea} {\begin{eqnarray}}
\newcommand{\eea} {\end{eqnarray}}
\newcommand{\ba} {\begin{array}}
\newcommand{\ea} {\end{array}}
\newcommand{\nn} {\nonumber}

\newtheorem{lemma}{Lemma}

 \title{Horizon  states and   the sign of their index
 in ${\cal N}=4$ dyons}
 \author{Aradhita Chattopadhyaya$^{\,a\,b}$, Justin R. David$^{\,c}$}
\affiliation{$^{\,a}$ School of Mathematics, Trinity College Dublin, Dublin 2, Ireland\\
$^{\,b}$ Hamilton Mathematical Institute, Trinity College, Dublin 2, Ireland\\
	$^{\,c}$Centre for High Energy Physics, Indian Institute of Science,\\
C. V. Raman Avenue, Bangalore 560012, India.}
\emailAdd{aradhita@maths.tcd.ie, justin@iisc.ac.in}

\abstract{Classical single centered
solutions of   $1/4$ BPS dyons  in ${\cal N}=4$ theories 
are usually constructed in  duality frames 
which contain non-trivial hair degrees of freedom localized outside the horizon. 
These modes are in addition to the fermionic zero modes associated with broken supersymmetry.
Identifying and removing the hair from the $1/4$ BPS index allows us to 
isolate the degrees of freedom associated with the horizon. 
The spherical symmetry of the horizon then ensures that index of the horizon states 
has to be positive. 
We  verify that this is indeed the case for the canonical example 
of dyons in type IIB theory on $K3\times T^2$ and prove this property holds for a class of 
states. 
We  generalise this observation to  all CHL orbifolds, this involves identifying the 
hair and isolating the horizon degrees of freedom. 
We then  identify the  horizon states 
for $1/4$  BPS dyons in ${\cal N}=4$ models obtained by freely acting 
${\mathbb{Z}}_2$ and ${\mathbb{ Z}}_3$ orbifolds of type IIB theory compactified on 
$T^6$ and observe that the index is again positive  for single centred black holes.
This observation coupled with the fact 
the $1/4$ 
BPS index of single centred solutions without removal of the hair violates  positivity 
indicates that there exists no duality frame in these models without non-trivial hair. 
}

\begin{document}
\maketitle
\flushbottom

\section{Introduction}
Counting  microscopic degrees of freedom for  extremal black holes in string 
theory  is a  useful probe into aspects of quantum gravity \cite{Strominger:1996sh}. 
For supersymmetric black holes,  one should in principle be able identify the  degrees of 
freedom both from the macroscopic solution as well as  count them from the microscopic
description of these black holes. 
The $1/4$ BPS dyonic black holes in ${\cal N}=4$ theory is a system which has been 
extensively studied in this context, see \cite{Sen:2007qy,Dabholkar:2012zz} for reviews. 
The identification of  the degrees of freedom   is complicated by the
fact that   classical solutions of black holes are multi-centered and usually they  contain 
hair degrees of freedom localized outside the horizon 
\cite{Sen:2009vz,Banerjee:2009uk,Jatkar:2009yd}. 
The microscopic  analysis counts all these configuations together. 
Let us make this precise, let $d_{\rm micro} (\vec q) $ be the 
degeneracy or in the case of the extremal supersymmetric black holes  the 
appropriate supersymmetric  index evaluated
from the microscopic description of a BPS state with charge $\vec q$. 
Similarly let $d_{\rm macro} (\vec q) $ be the  corresponding macroscopic index. 
Then
\begin{eqnarray}\label{split}
d_{\rm macro}( \vec q) = 
\sum_{n} \sum_{ \stackrel{ \{\vec q_i \} , \vec q_{\rm hair}}
{\sum_{i=1}^n( \vec q_i + \vec q_{\rm hair} ) = \vec q}  }
\left( \prod_{i = 1}^n d_{\rm hor} ( \vec q_i ) \right) d_{\rm hair} ( \vec q_{\rm hair} ; \{\vec q_i \} )
\end{eqnarray}
Each term on the right hand side of (\ref{split})  is the contribution  to the index 
of say the $n$-centered black hole configuration. $d_{\rm hor} ( \vec q_i)$ is the 
contribution to the index from the horizon degrees of freedom with charge $q_i$ and 
$d_{\rm hair} ( \vec q_{\rm hair} ; \{\vec q_i \} )$  is the index of the hair carrying total 
charge $\vec q_{hair}$  of a $n$-centered black hole whose horizons carry charges 
$\vec q_1, \cdots \vec q_n$. 
We expect 
\begin{eqnarray}\label{basic}
d_{\rm macro} (\vec q) = d_{\rm micro}( \vec q) .
\end{eqnarray}
It would simplify matters if we can restrict our attention to single centred black hole configurations. 
Then  (\ref{split}) indicates that we would need to identify the hair to isolate the horizon degrees 
of freedom.  Since we are dealing with $1/4$  BPS states in ${\cal N}=4$ theories, which break
12 supersymmetries,   the degeneracy  $d(\vec q)$ will refer to the index 
\begin{equation}\label{heltrace}
B_6 = \frac{1}{6!}{\rm Tr} ( (2J)^6 (-1)^{2J}  ) ,
\end{equation}
where $J$ is  the   component of the angular momentum in say the $3$ direction.  The factorized 
form of the Hilbert space corresponding to the hair degrees of freedom and the horizon degrees of freedom follows from the fact the these are well separated due the presence of an infinite throat
\cite{Sen:2009vz}. 

The utility of identifying the horizon degrees of freedom lies in the fact that 
the horizon is spherically symmetric and therefore carries zero momentum $J=0$. 
The index taken over the 
horizon states  reduces to $(-1)^{2J} d_{\rm hor} = d_{\rm hor}$, where $d_{\rm hor}$ is the total 
number of states associated with the horizon.  
Therefore the index of the horizon states must be a positive number. 
This leads to  an important check on the microscopic counting and the equality (\ref{basic}).
Once one determines the hair degrees of freedom for a given macroscopic black hole 
and factors them out of the index, what must remain is a positive number which counts 
the index of the horizon states. 
This argument clearly relies on what are the hair degrees of freedom and this in turn depends 
on the  duality frame of  the macroscopic solution.
This prediction was tested in \cite{Sen:2010mz} with the assumption that there exists a frame 
in which the only hair  degrees of freedom are the fermionic zero modes associated with 
the broken supersymmetry generators. 
For black holes in ${\cal N}=8$ there is evidence towards this fact in 
\cite{Chowdhury:2014yca,Chowdhury:2015gbk}.  These authors 
worked in  
a frame in which the black hole configuration reduced to a system of only  D-branes and showed 
 the only hair degrees of freedom were the fermionic zero modes 
and the BPS configuration indeed had zero angular momentum. 
However such a frame has not  yet been shown to exist for black holes in ${\cal N}=4$ theory. 

Given this situation, one way of proceeding  is to evaluate the partition functions corresponding 
to the hair degrees of freedom and isolate the horizon degrees of freedom in a given frame. 
This has already been done in \cite{Banerjee:2009uk,Jatkar:2009yd}, 
for  $1/4$ BPS dyons in the type IIB frame, 
 but a test of positivity of the index for the resulting 
horizon degrees of freedom has not been done. 
We perform this analysis in this paper and indeed demonstrate that 
the $d_{\rm hor}$ is indeed positive. This is quite remarkable as we will see, since 
factorizing the hair degrees of freedom naively seems to introduce terms with negative contributions
to the index.  We  adapt the proof of \cite{Bringmann:2012zr}  for configurations with 
magnetic charge $P^2 =2$ and demonstrate that the  index is positive. 
We then extend this observation to all the CHL models and to other orbifolds associated 
with Mathieu moonshine introduced in \cite{Chattopadhyaya:2017ews,Persson:2013xpa}. 

In \cite{Chattopadhyaya:2017ews,Chattopadhyaya:2018xvg} 
it was observed that for ${\cal N}=4$ models obtained by freely acting 
$\mathbb{Z}_2, \mathbb{Z}_3$ orbifolds  of type IIB on  $T^6$, 
the index  for single centered configurations after
factorising the sign due to the fermionic zero modes did not obey the expectation $d_{\rm hor}$
is positive \footnote{Please see tables \ref{qp0}, \ref{qp1}, \ref{qp2}  reproducing this observation }. 
But as the above discussion shows,  a possible reason for this could be that the assumption that there
exists a frame in which the fermionic zero modes are the only hair degrees of freedom 
might not be true. 
Therefore we re-examine this question in this paper. 
Following the same procedure used in the CHL models we isolate the hair degrees of freedom 
in the type IIB frame. Then on examining the sign of the index for single centered black holes
we observe that $d_{\rm hor}$  is positive.

The organisation of the paper is as follows. 
In the  section  \ref{sechorstate} we briefly review the statements about the hair 
and the the partition function for the horizon degrees of freedom. 
for the $1/4$ BPS dyonic black hole in type IIB compactified on $K3\times T^2$. 
We then generalise this to all the CHL orbifolds as well as other orbifolds 
associated with Mathieu Moonshine. 
Finally we construct the partition function for the horizon states for the 
toroidal models obtained by freely acting $\mathbb{Z}_2, \mathbb{Z}_3$ orbifolds of
type IIB on $T^6$. 
In section \ref{check}  we perform a consistency check on the $d_{\rm hor}$ obtained. 
This check relies on the fact that the 5-dimensional BMPV black hole has the same near 
horizon geometry \cite{Banerjee:2009uk,Jatkar:2009yd}. Therefore 
$d_{\rm hor}$ for the BMPV black hole should agree with that of the $1/4$ BPS dyon. 
We show that this is indeed the case for all the examples. 
Finally in section \ref{signind}, armed with the $d_{\rm hor}$ for all the models we study the 
positivity of the index for single centered black holes  for all the models. 
We have evaluated numerically 
the indices of horizon states   for several charges in all  the ${\cal N}=4$ models for which 
dyon partition functions are known 
which confirm that  the index is positive. 
We  adapt the proof of \cite{Bringmann:2012zr} to show that the index is positive for charge configurations
with $P^2 = 2$. 
Section  \ref{conclusions}  contains our conclusions.

\section{Horizon states for the $1/4$ BPS dyon} \label{sechorstate}

In this section we construct the partition function for the horizon states for $1/4$ BPS dyons in 
${\cal N}=4$ compactifications. This is done by  identifying the `hair' degrees of freedom which 
are localized outside the horizon.
Such a partition function for the horizon states
 was constructed for the canonical ${\cal N}=4$ theory obtained by 
compactifying type IIB  string theory  $K3\times T^2$ in 
\cite{Banerjee:2009uk,Jatkar:2009yd} in the type IIB frame. 
We review this in section  and then extend the analysis for other ${\cal N}=4$ models. 

The  ${\cal N}=4$ compactifications of interest are type IIB theory on $K3\times T^2/\mathbb{Z}_N$
where $\mathbb{Z}_N$ acts as an automorphisim  $g'$ on $K3$ along with a 
shift of $1/N$ units on one of the circles of $T^2$.  The action of $g'$ can be labelled by the 
$26$ conjugacy classes of the Mathieu group $M_{23}$.  The classes $pA$ with
$ p =2, 3,  5, 6, 7, 8  $  and the class $4B$ are called as Nikulin's automorphism of $K3$. 
They were first introduced in \cite{Chaudhuri:1995ve,Chaudhuri:1995dj} 
as models dual to heterotic string theory with 
${\cal N}=4$ superysmmetry but with gauge groups with reduced from the maximal rank of $28$.
All these compactifications admit $1/4$ BPS dyons,  let $(Q, P)$  be the electric and magnetic 
charge vector of these dyons, then the $1/4$ BPS index $B_6$ is given by 
\cite{Dijkgraaf:1996it,Jatkar:2005bh,David:2006ji,David:2006yn,David:2006ud}
%After removal of these modes we observe that for $Q\cdot P\ge 0$ we always have a positive sign on for the horizon degrees of freedom, where
%The degeneracy of single centered black holes are given by $-B_6$ defines as follows:   
\be \label{b6phi10}
-B_6=\frac{1}{N}(-1)^{Q \cdot P +1}\int_{{\cal C}}{ d}\rho{ d}
\sigma {d}v\; e^{-\pi i ( N \rho Q^2+\sigma P^2/N +2v Q\cdot P)}\frac{1}{\tilde\Phi_k(\rho,\sigma, v)},
\ee
where ${\cal C}$ is a contour in the complex 3-plane defined by 
\begin{eqnarray}\label{contour}
\rho_2 = M_1, \qquad \sigma_2 = M_2, \qquad \tilde v_2 = - M_3, \\ \nonumber
0\leq \rho_1 \leq 1, \qquad 0 \leq \tilde\sigma_1 \leq N, \qquad 0 \leq \tilde v_1 \leq 1.
\end{eqnarray}
Here $\rho= \rho_1 + i \rho_2, \sigma= \sigma_1 + i \sigma_2, v = v_1 + i v_2$ and 
 $M_1, M_2, M_3$ are  positive numbers, which are fixed and large and
$M_3 << M_1, M_2$.  The contour in (\ref{contour}) implies that we first expand
in powers or $e^{2 \pi i \rho}, e^{2\pi i \sigma}$ and at the end perform the expansion 
in $e^{2\pi i v}$.

The Siegel modular  form  $\tilde\Phi(\rho, \sigma, v) $ transforming under $Sp(2,\mathbb{Z})$, 
or its subgroups  for $N>1$  admits an infinite product representation given by 
\bea \label{phi10}
\tilde{\Phi}_k(\rho,\sigma,v)&=&e^{2\pi i(\rho+ \sigma/N+v)} \times \\ \nn
&& \prod_{r=0}^N\prod_{\begin{smallmatrix}k\in \mathbb{Z}+\frac{r}{N},\;l\in \mathbb{Z},\\ j\in \mathbb{Z}\\ k',l\geq 0;\;\; j<0 \;k'=l=0\end{smallmatrix}}(1-e^{2\pi i(k'\sigma+l\rho+jv)})^{\sum_{s=0}^{N-1} c^{(r,s)}(4kl-j^2)}.
\eea
The coefficients $c^{(r,s)}$  are determined from  the expansion of the twisted elliptic genera for 
the various order $N$ orbifolds  $g'$  of $K^3$. 
The twisted elliptic genus of $K3$ is defined by 
\begin{eqnarray}
F^{(r, s) }( \tau, z) &=& \frac{1}{N} {\rm Tr}_{RR \;\; g^{\prime r} } 
\left[ ( -1)^{F_{K3} + \bar F_{K3} } g^{\prime s} e^{2\pi i z F_{K3} }
q^{L_0 - \frac{c}{24} }\bar q ^{\bar L_0 - \frac{\bar c}{24} } \right],   \\ \nonumber
&=& \sum_{ j \in \mathbb{Z} , \;  n \in \mathbb{Z}/N}
c^{( r, s) } ( 4n - j^2) e^{2\pi i n \tau + 2\pi i j z}. \\\ \nonumber
& & \qquad\qquad \qquad \qquad 0\leq r, s\leq N-1.
\end{eqnarray}
The trace is performed over the Ramond-Ramond sector of the ${\cal N}=(4, 4) $  super conformal 
field theory of $K3$ with $( c, \bar c) = ( 6, 6 )$, $F$ is the Fermion number and $j$ is the left
moving  $U(1)$ charge   of the $SU(2)$ $R$-symmetry of $K3$. 
The twisted elliptic genera for the  $g'$ corresponding to conjugacy classes of $M_{23}\cup M_{24}$
have been evaluated in \cite{Chattopadhyaya:2017ews}. 
These take the form
\begin{eqnarray}\label{explielip}
F^{(0, 0)} ( \tau, z) &=& \alpha_{g'}^{(0, 0)} A( \tau, z) ,  \\ \nonumber
F^{(r, s) } ( \tau, z) &=& \alpha_{g'}^{(r, s) } A( \tau, z) + \beta_{g'}^{(r, s) }(\tau)  B(\tau, z) , 
\\ \nonumber
 && \qquad\qquad r, s \in \{0, 1, \cdots N-1 \} \; \hbox{with} ( r, s) \neq (0, 0), 
\end{eqnarray}
where
\begin{eqnarray}\label{ab}
A(\tau, z) &=& \frac{\theta_2^2(\tau,z)}{\theta_2^2(\tau, 0) }+\frac{\theta_3^2(\tau,z)}{\theta_3^2(\tau, 0)}+\frac{\theta_4^2(\tau,z)}{\theta_4^2(\tau, 0) }, \\ \nonumber
B(\tau, z) &=& \frac{\theta_1^2(\tau, z) }{\eta^6(\tau) }.
\end{eqnarray}
The coefficients $\alpha_{g'}^{(r, s)}$ in (\ref{explielip})  are numerical constants, while 
$\beta_{g'}^{(r, s) }(\tau) $ are modular forms that transform under $\Gamma_0(N)$. 
For $g'$ corresponding to conjugacy classes of $M_{23}$, they can be read out 
from  appendix E of \cite{Chattopadhyaya:2017ews}. 
For example, in the case of the $pA$ orbifolds with $p =1,2, 3, 5, 7$, they are given by 
\cite{David:2006ji}. 
\begin{eqnarray}\label{2atwist}
F^{(0, 0)} &=& \frac{8}{N} A(\tau, z) , \\ \nonumber
F^{(0, s)} &=& \frac{8}{(N+1)N} A(\tau, z) - \frac{2}{N+1} B(\tau, z) {\cal E}_N(\tau) , \\ \nonumber
F^{(r,rk)} &=& \frac{8}{N(N+1)} A(\tau, z) + \frac{2}{N(N+1)} B(\tau, z) {\cal E}_N(\frac{\tau+k}{N} ) , 
\\ \nonumber
{\cal E}_N(\tau) &=& \frac{12i}{\pi ( N-1) } \partial_\tau [ \ln \eta(\tau) - \ln \eta( N\tau) ]. 
\end{eqnarray}
For $N$ composite  corresponding to the 
classes $4B, 6A, 8A$, the strategy for construction of the twisted elliptic genus 
was first given in \cite{Govindarajan:2009qt}
 and it was worked out explicitly for the $4B$ example 
 \footnote{Suresh Govindarajan  informed us  that the authors of \cite{Govindarajan:2009qt}
 also explicitly constructed all the sectors of 
  the $6A$ and $8A$ twisted elliptic genera though it was not reported in the paper.}.  
  The papers \cite{Cheng:2010pq,Eguchi:2010fg,Gaberdiel:2010ch} contain the twining characters, 
 $F^{(0, s)}$ and 
 \cite{Gaberdiel:2012gf}  also contains the 
 strategy to construct the twisted elliptic genera
  for other conjugacy classes of $M_{23}$ and a  Mathematica code 
 for generating the elliptic genera. 

The weight of the Siegel modular form $\tilde \Phi( \rho,\sigma, v) $ is given by 
\begin{eqnarray}
k = \frac{1}{2} \sum_{s=0}^{N-1} c^{(0, s)} ( 0 ) .
\end{eqnarray}
For the classes $pA, p = 1, 2, 3, 5, 7, 11$ we have
\begin{equation}
k = \frac{24}{ p +1} - 2,
\end{equation}
for $4B, 6A, 8A$  we have $k = 3, 2, 1$ respectively and for $14A, 15A$ $k=0$. 

Finally, as discussed  in the introduction the study of horizon states would be much 
simpler if one could focus on single centered dyons.  Such a system would 
have only one horizon. 
The choice of the 
contour chosen in (\ref{contour}) together with 
some kinematic  constraints on charges such as (\ref{keres}) ensures that 
we are in  the attractor region of the axion-dilaton moduli
and the index given by (\ref{b6phi10}) is that of single centred dyons \cite{Sen:2007vb,Sen:2010mz}. 
All the indices evaluated in this section paper is done using the contour (\ref{contour}).

\subsection{The canonical example: $K3\times T^2$}

In the work of \cite{Banerjee:2009uk,Jatkar:2009yd} the hair modes  of the $1/4$ BPS dyonic 
black hole in  type IIB theory compactified on $K3\times T^2$ were constructed.  
Here we briefly review this construction.
These modes were shown to be deformations localized outside the horizon and they preserved
supersymmetry. 
Let us first recall that the dyonic black hole in 4-dimensions is constructed by placing the 
$5$ dimensional BMPV black hole or the rotating  D1-D5 system  \cite{Breckenridge:1996is}
 in Taub-Nut space \cite{Gaiotto:2005gf}. The Taub-Nut space 
has the geometry which at the origin is $R^4$ but  at infinity it is  $R^3\times \tilde S $. 
The isometry along  $S^1$ coincides with the angular direction the BMPV rotates.
The hair modes arise from the collective modes of the D1-D5 system thought of as
an effective string along say the $x^5$ and the time  $t$ directions.
Therefore these modes are oscillations of the effective string, they 
are left moving since they have to preserve supersymmetry \footnote{It is easy to see from the heterotic frame that only left moving oscillations preserve supersymmetry.} 
After allowing the fermionic zero modes associated with the $12$ broken susy generators 
to saturate $(2J)^6/6!$ in the helicity trace given in (\ref{heltrace}), the 
non-trivial hair modes consist of 
\begin{itemize}
	\item 4 left moving fermionic modes  arising from the deformations of the 
	gravitino  giving rise to the contribution 
	\begin{equation}
	Z_{{\rm hair}:1A}^{ 4d: f }  = \prod_{l=1}^\infty  (  1- e^{2\pi i l \rho} )^4
	\end{equation}
	\item 3 left moving bosonic modes associated with the oscillation of the effective
	 string in the 3 transverse directions 
	 $\mathbb{R}^3$ as Taub-NUT is assymptotically $\mathbb{R}^3\times \tilde S^1$.
	\begin{equation}
	Z_{{\rm hair}: 1A}^{ 4d: { \perp} }  =  \prod_{l=1}^\infty  \frac{1}{( 1- e^{ 2\pi i l \rho} ) ^3}
	\end{equation}
	\item 21 left moving bosonic modes, these arise from the  deformations of the 
	21 anti-self-dual  forms of type IIB on K3. These deformations involve $21$ scalar functions
	folded with the $2$ form $d\omega_{TN}$ on the Taub-Nut given by 
	\begin{equation}\label{defhv}
	\delta H^s = h^s ( v) dv \wedge d\omega _{TN}, \qquad v = t + x^5, \quad s = 1, \cdots 21
	\end{equation}
	Counting these oscillations we obtain
	\begin{equation}
	Z_{{\rm hair}: 1A}^{4d:{\rm asd} } = \prod_{l=1}^\infty  \frac{1}{( 1- e^{ 2\pi i l \rho} ) ^{21}}.
	\end{equation}
	The $21$ anti-self dual forms arise from compactifying the RR 4-form on the 
	$19$ anti-self dual $2$ form of the $K3$ together with the  NS 2-form and the RR 2-form 
	of type IIB. 
\end{itemize}
Note that in the partition function we labelled the chemical potential to count the osciallations
by $\rho$, this is because exciting these  left moving momentum modes   
correspond to exciting the electric charge 
of the dyon \cite{David:2006yn}.  
Now combining these partition functions we obtain 
\begin{eqnarray} \label{k3h}
Z_{ {\rm hair } : 1A}^{4d} &=& Z_{{\rm hair}:1A}^{ 4d: f }\times Z_{{\rm hair}: 1A}^{ 4d: { \perp} } 
\times Z_{{\rm hair}: 1A}^{4d:{\rm asd} }  \\ \nonumber
& =& \prod_{l=1}^{\infty}(1-e^{2\pi i (l\rho)})^{-20}.
\end{eqnarray} 
 The Bosonic  hair partition function  is given by 
 \begin{equation}
 Z_{{\rm hair}: 1A }^{ \rm bosons}= Z_{{\rm hair}: 1A}^{ 4d: { \perp} } 
\times Z_{{\rm hair}: 1A}^{4d:{\rm asd} } 
  = \frac{ e^{2\pi i \rho}}{ \eta^{24} ( \rho) }, 
 \end{equation}
 this  is identical to that of the 
 counting the degeneracy of purely electric states  in  this model without the zero point energy. 
 This observation will help in the generalizations to CHL models.

To obtain the partition function of horizon states we factor out the hair degrees of freedom 
resulting in 
\begin{equation}
Z_{{\rm hor}} = \frac{ 1}{ \Phi_{10} ( \rho, \sigma, v)  Z_{ {\rm hair } : 1A}^{4d}  }.
\end{equation}
 The  index for the horizon states  can be
 then be obtained   by extracting the Fourier coefficients  using the expression given by 
 \be \label{dhork3}
d_{hor}=-(-1)^{Q \cdot P}\int_{{\cal C}}{ d}\rho{ d}
\sigma { d}v\; e^{-\pi i (\rho Q^2+\sigma P^2+2v Q\cdot P)}\frac{1}{\tilde\Phi_{10}(\rho,\sigma, v)}\prod_{l=1}^{\infty}(1-e^{2\pi i (l\rho)})^{20}.
\ee
Here the contour ${\cal C}$ is same as that defined in (\ref{contour}). 

\subsection{Orbifolds of $K3\times T^2$}

\subsubsection*{$2A$ orbfiold}

Before we present the analysis for the most general orbifold, let us examine in detail 
the analysis for the $2A$ orbifold.  In this case, the orbifold acts by exchanging $8$ pairs
of anti-self dual  $(1, 1)$ forms out of the $19$ anti-self dual forms of $K3$ with the $1/2$ shift 
on $S^1$ \cite{Chaudhuri:1995dj}. 
Note that because of the $1/2$ shift, the natural unit of momentum on $S^1$ is $N=2$. 
With this input we are ready to repeat the analysis for the partition function of the hair modes
\begin{itemize}
\item 
The 4 left moving fermionic modes  arising from the deformations of the 
	gravitino  give rise to the contribution 
	\begin{equation}
	Z_{{\rm hair}:2A}^{ 4d: f }  = \prod_{l=1}^\infty  (  1- e^{4\pi i l \rho} )^4.
	\end{equation}
	Note that due to the fact that the periodicity is now $\frac{2\pi }{N}$, the unit of momentum 
	is doubled. 
\item The $3$ transverse bosonic deformations along $R^3$   of the effective string  results in 
\begin{equation}
Z_{{\rm hair}: 2A}^{ 4d: { \perp} }  =  \prod_{l=1}^\infty  \frac{1}{( 1- e^{ 4\pi i l \rho} ) ^3}.
\end{equation}
\item The action of the orbifold projects out $8$ anti-self dual forms. The  analysis for $13 = 11 + 2$.
\footnote
{The 2 arises from the anti-self dual component of the RR 2-form and the NS 2-form.}
 invariant 
anti-self dual forms  proceeds as before except for the fact that the unit of momentum is $2$
\begin{equation}
Z_{{\rm hair}: 2A}^{4d:{\rm asd} } |_{\rm invariant} = \prod_{l=1}^\infty  \frac{1}{( 1- e^{ 4\pi i l \rho} ) ^{13}}.
\end{equation}
Consider the following boundary conditions of the function $h(s)$
in (\ref{defhv})  for the $8$ projected 
anti-self dual forms. 
\begin{equation}
h(v + \frac{2\pi}{N} ) = - h( v) , \qquad \qquad    N=2.
\end{equation}
These deformation pick up sign when one move by $1/2$ unit on $S^1$. The partition function 
corresponding to these modes is given by 
\begin{equation}
 Z_{{\rm hair}: 2A}^{4d:{\rm asd} } |_{\rm twisted} = 
 \prod_{l=1}^\infty  \frac{1}{( 1- e^{ 2\pi i (2l  -1) \rho} ) ^{8}}.
 \end{equation}
 Note that these modes are twisted for the circle of radius $2\pi/N, N =2$, they obey 
 anti-periodic  boundary conditions. However in  supergravity 
 periodicities are measured over  the circle of radius $2\pi $  and they are periodic  for this radius, 
  therefore these modes  can be  counted as hair modes. 
 Together, the contribution of the anti-self dual forms to the partition function 
 is given by 
 \begin{eqnarray}
 Z_{{\rm hair}: 2A}^{4d:{\rm asd} } &=& Z_{{\rm hair}: 2A}^{4d:{\rm asd} } |_{\rm invariant} 
 \times Z_{{\rm hair}: 2A}^{4d:{\rm asd} } |_{\rm twisted} , \\ \nonumber
 &=&  \prod_{l=1}^\infty \frac{1}{( 1- e^{ 4\pi i l \rho})^5}
  \prod_{l=1}^\infty  \frac{1}{( 1- e^{ 2\pi i l \rho} ) ^{8}}
  \end{eqnarray}
  \end{itemize}
  Now combining all the hair modes we obtain
  \begin{eqnarray}
  Z_{ {\rm hair } : 2A}^{4d} &=& Z_{{\rm hair}:2A}^{ 4d: f }\times Z_{{\rm hair}: 2A}^{ 4d: { \perp} } 
\times Z_{{\rm hair}: 2A}^{4d:{\rm asd} }  \\ \nonumber
& =& \prod_{l=1}^{\infty}(1-e^{4\pi i l\rho)})^{-4} ( 1 - e^{2\pi i l \rho} )^{-8} .
\end{eqnarray} 
Observe that the  partition function of the bosonic hair modes is given by 
\begin{eqnarray}
 Z_{{\rm hair}: 2A }^{ 4d:\;  b} &=& Z_{{\rm hair}: 2A}^{ 4d: { \perp} } 
\times Z_{{\rm hair}: 2A}^{4d:{\rm asd} } , \\ \nonumber
&=& \prod_{l =1}^\infty ( 1- e^{ 4\pi i l \rho } )^{-8} ( 1- e^{2\pi i l \rho} )^{-8} , \\ \nonumber
&=& \frac{e^{2\pi i \rho} }{ \eta^8 ( 2 \rho)\eta^{8} ( \rho) } .
\end{eqnarray}
This is the partition  function of the fundamental string in the $N=2$ CHL orbifold of the 
heterotic theory  with the zero point energy removed \cite{Dabholkar:2005by,David:2006yn}.

\subsubsection*{$pA$ orbifolds $p = 2, 3, 5, 7$}

The construction of the hair modes for the case of orbifolds of prime order, the 
method proceeds as discussed in detail for the $2A$ orbifold. 
In each case we need to count the number of $2$-forms which are left invariant and which pick up phases
and evaluate the partition function.
The result for the bosonic hair modes is given by
\be\label{bosonic}
Z_{{\rm hair} : \; pA }^{4d: \; b}
=\prod_{l=1}^{\infty}\frac{1}{(1-e^{2\pi i \rho N l})^{k+2}(1-e^{2\pi i l\rho})^{k+2}}.
\ee
where 
\begin{equation}
k = \frac{24}{ p +1} - 2.
\end{equation}
Note that this is the partition of the states containing only the electric charges or the fundamental 
string without the zero point energy \cite{David:2006yn}. 
Now including the $4$ fermionic deformations we obtain 
\begin{eqnarray}\label{chlh}
Z_{{\rm hair}:\;  pA}^{4d} &=&\prod_{l=1}^{\infty}(1-e^{2\pi i (Nl\rho)})^{-(k+2)}(1-e^{2\pi i (l\rho)})^{-(k+2)}(1-e^{2\pi i (Nl\rho)})^{4}\\ \nn
&=& \prod_{l=1}^{\infty}(1-e^{2\pi i (Nl\rho)})^{-{2k}}\prod_{N\nmid l}(1-e^{2\pi i (l\rho)})^{-(k+2)}
\end{eqnarray}
It is useful to rewrite this expression as follows
\begin{eqnarray}\label{chlall}
Z_{{\rm hair}:\; pA}^{4d}&=&\prod_{l=1}^{\infty}
(1-e^{2\pi i (Nl\rho)})^{-\sum c^{(0,s)}(0)}\prod_{N\nmid l}(1-e^{2\pi i (l\rho)})^{-\sum e^{-2\pi i sl/N}c^{(0,s)}(0)}\\ \nn
&=& \prod_{l\ne 0}(1-e^{2\pi i (l\rho)})^{-\sum e^{-2\pi i sl/N}c^{(0,s)}(0)}
\end{eqnarray}
The sum is on the range of $s=0$ to $N-1$ and $N\nmid l$ implies $N$ does not divide $l$.

The values of $\sum_{s=0}^{N-1}e^{-2\pi i s l/N}c^{(0,s)}(-b^2)$ for prime $N$ are listed in 
table \ref{tablep}
\begin{table}[H]
	\renewcommand{\arraystretch}{0.5}
	\begin{center}
		\vspace{0.5cm}
		\begin{tabular}{|c|c|c|c|}
			\hline
			& & &  \\
			$N$ & $l$ & $-b^2$ & $\sum_{s=0}^{N-1}e^{-2\pi i sl/N}c^{(0,s)}(-b^2)$\\ 
			& & &  \\
		\hline
			& & & \\
$p$ & $N|l$ & 0 & $2k=\frac{48}{N+1}-4$\\
 & & & \\
 & & $-1$ & 2\\
 	\cline{2-4}
  & & & \\
  & $N\nmid l$ & 0 & $k+2=\frac{24}{N+1}$\\
& & & \\
  & & $-1$ & 0\\
\hline			
\end{tabular}
\end{center}
\vspace{-0.5cm}
\caption{Values of  $\sum_{s=0}^{N-1}e^{-2\pi i s l/N}c^{(0,s)}(-b^2)$ for orbifolds of $K3$ with  prime
order \\ $(N=p)$} \label{tablep}
\renewcommand{\arraystretch}{0.5}
\end{table}

\subsubsection*{Orbifolds of composite order: $4B, 6A, 8A$ }

One can count the hair modes in a similar fashion as the  orbifolds of prime order.
The only difference would arise for the bosonic modes 
$Z_{\rm hair}^{\rm bosons}$,  which needs to be replaced by the fundamental string in these theories without the zero point energy. Including the  4 fermionic hairs, we see that the answer 
can be written in the same form as that seen for orbifolds with prime order. 
\begin{eqnarray}
Z_{{\rm hair}: {\rm CHL}}^{4d} &=&\prod_{l=1}^{\infty}(1-e^{2\pi i (Nl\rho)})^{-\sum c^{(0,s)}(0)}\prod_{N\nmid l}(1-e^{2\pi i (l\rho)})^{-\sum e^{-2\pi i sl/N}c^{(0,s)}(0)}.
\end{eqnarray}
The sum  ranges from  $s=0$ to $N-1$.
This can be  rewritten as 
\begin{eqnarray}
Z_{{\rm hair}: {\rm CHL}}^{4d}&=&\prod_{l=1}^{\infty}(1-e^{2\pi i (l\rho)})^{-\sum e^{-2\pi i sl/N}c^{(0,s)}(0)}.
\end{eqnarray}

For the geometric CHL orbifolds ,  we list 
$\sum_{s=0}^{N-1}e^{-2\pi i s l/N}c^{(0,s)}(-b^2)$ for different $N=4,6,8$ in table \ref{tablenonp}
\begin{table}[H]
	\renewcommand{\arraystretch}{0.5}
	\begin{center}
		\vspace{0.5cm}
		\begin{tabular}{|c|c|c|c|}
			\hline
			& & &  \\
			$N$ & $l$ & $-b^2$ & $\sum_{s=0}^{N-1}e^{-2\pi i sl/N}c^{(0,s)}(-b^2)$\\ 
			& & &  \\
			\hline
			& & & \\
			$4$ & $4|l$ & 0 & $6$\\
			& & & \\
			& & $-1$ & 2\\
			\cline{2-4}
			& & & \\
			 & $2| l,\; 4\nmid l$ & 0 & 6\\
			& & & \\
%			& & $-1$ & 0\\
			\cline{2-4}
				& & & \\
			& $2\nmid l$ & 0 & 4\\
			& & & \\
%			& & $-1$ & 0\\
			\hline  \hline	
				& & & \\
			$6$ & $6|l$ & 0 & $4$\\
			& & & \\
			& & $-1$ & 2\\
			\cline{2-4}
			& & & \\
			& $2| l,\; 6\nmid l$ & 0 & 4\\
			& & & \\
%			& & $-1$ & 0\\
			\cline{2-4}
				& & & \\
			& $3| l,\; 6\nmid l$ & 0 & 4\\
			& & & \\
%			& & $-1$ & 0\\
			\cline{2-4}
			& & & \\
			& $2\nmid l, 3\mid l$ & 0 & 2\\
			& & & \\
%			& & $-1$ & 0\\
			\hline	\hline		
						& & & \\
				$8$ & $8|l$ & 0 & $2$\\
				& & & \\
				& & $-1$ & 2\\
				\cline{2-4}
				& & & \\
				& $2| l,\; 4\nmid l$ & 0 & 3\\
				& & & \\
%				& & $-1$ & 0\\
				\cline{2-4}
				& & & \\
				& $4| l,\; 8\nmid l$ & 0 & 4\\
				& & & \\
%				& & $-1$ & 0\\
				\cline{2-4}
				& & & \\
				& $2\nmid l$ & 0 & 2\\
				& & & \\
%				& & $-1$ & 0\\
				\hline			
		\end{tabular}
	\end{center}
	\vspace{-0.5cm}
	\caption{Values of  $\sum_{s=0}^{N-1}e^{-2\pi i s l/N}c^{(0,s)}(-b^2)$ for non-prime CHL orbifolds of $K3$. $\sum_{s=0}^{N-1}e^{-2\pi i s l/N}c^{(0,s)}(-1)=0$ if $N\nmid l$ for any of these cases.} \label{tablenonp}
	\renewcommand{\arraystretch}{0.5}
\end{table}
Using the  data from table \ref{tablenonp} we obtain
\begin{eqnarray}\label{chlcom}
Z_{{\rm hair}: 4B}^{4d}  &=&\prod_{l=1}^{\infty}(1-e^{2\pi i (4l\rho)})^{4}
(1-e^{2\pi i (4l\rho)})^{-4}(1-e^{2\pi i (2l\rho)})^{-2}(1-e^{2\pi i (l\rho)})^{-4}\\ \nn
&=& \prod_{l=1}^{\infty}(1-e^{2\pi i (2l\rho)})^{-2}(1-e^{2\pi i (l\rho)})^{-4}\\ \nn
Z_{{\rm hair}: 6A}^{4d}&=&\prod_{l=1}^{\infty}(1-e^{2\pi i (6l\rho)})^{4}
(1-e^{2\pi i (6l\rho)})^{-2}(1-e^{2\pi i (2l\rho)})^{-2}(1-e^{2\pi i (3l\rho)})^{-2}(1-e^{2\pi i (l\rho)})^{-2}\\ \nn
&=&\prod_{l=1}^{\infty}(1-e^{2\pi i (6l\rho)})^{2}
(1-e^{2\pi i (2l\rho)})^{-2}(1-e^{2\pi i (3l\rho)})^{-2}(1-e^{2\pi i (l\rho)})^{-2}\\ \nn
Z_{{\rm hair}: 8A}^{4d}&=&\prod_{l=1}^{\infty}(1-e^{2\pi i (8l\rho)})^{4}
(1-e^{2\pi i (8l\rho)})^{-2}(1-e^{2\pi i (2l\rho)})^{-1}(1-e^{2\pi i (4l\rho)})^{-1}(1-e^{2\pi i (l\rho)})^{-2}\\ \nn
&=&\prod_{l=1}^{\infty}(1-e^{2\pi i (8l\rho)})^{2}
(1-e^{2\pi i (2l\rho)})^{-1}(1-e^{2\pi i (4l\rho)})^{-1}(1-e^{2\pi i (l\rho)})^{-2}.
\end{eqnarray}

\subsubsection*{Horizon states}

We factor out the hair degrees of freedom to obtain the horizon states, this is given by 
\be
Z_{{\rm hor}:{\rm CHL} }^{4d}=-\frac{1}{\tilde\Phi_k(\rho,\sigma, v)}\prod_{l=1}^{\infty}(1-e^{2\pi i (l\rho)})^{\sum_s e^{-2\pi i sl/N}c^{(0,s)}(0)}
\ee
It is useful to  use the product form of $\tilde\Phi_k$ given in (\ref{phi10}) to 
rewrite the partition function of the horizon states as follows
\begin{eqnarray} \label{horparti}
Z_{{\rm hor}:{\rm CHL} }^{4d} &=&
-e^{-2\pi i(\rho+\sigma/N+v)} \prod_{r=0}^{N-1}\prod_{\begin{smallmatrix}k'\in \mathcal{Z}+r/N,l\in \mathcal{Z},\\ j\in \mathcal{Z}\\ k'>0,l\ge 0\end{smallmatrix}}(1-e^{2\pi i(k'\sigma+ l\rho+jv)})^{-\sum_s e^{-2\pi isl/N}c^{(r,s)}(4k'l-j^2)} \nn \\
&&\times \prod_{l=1}^{\infty}(1-e^{2\pi i (Nl\rho+v)})^{-2}\prod_{l=1}^{\infty}(1-e^{2\pi i (Nl\rho-v)})^{-2} (1-e^{-2\pi i v})^{-2}\\ \nn
&=&-e^{-2\pi i(\rho+\sigma/N)} 
\prod_{r=0}^{N-1} \prod_{\begin{smallmatrix}k'\in \mathcal{Z}+r/N,l\in \mathcal{Z},\\ j\in \mathcal{Z}\\ k' >0,l\ge 0\end{smallmatrix}}(1-e^{2\pi i(k'\sigma+ l\rho+jv)})^{-\sum_s e^{-2\pi isl/N}c^{(r,s)}(4k'l-j^2)}\\ \nn
&&\times \prod_{l=1}^{\infty}(1-e^{2\pi i (Nl\rho+v)})^{-2}\prod_{l=1}^{\infty}(1-e^{2\pi i (Nl\rho-v)})^{-2} (e^{\pi i v}-e^{-\pi i v})^{-2}.
\end{eqnarray}
This form of the horizon partition function is useful in the  next section. 
 The index for the  horizon states is given by 
\begin{eqnarray}
d_{{\rm hor}:{\rm CHL} }&=& -(-1)^{Q \cdot P}\int_{{\cal C}}{d}\rho{ d}
\sigma { d}v\; e^{-\pi i (N\rho  Q^2+\sigma P^2/N+2v Q\cdot P)}  \frac{1}{\tilde\Phi_k(\rho,\sigma, v)}\times  \nonumber \\
& & \qquad\qquad\qquad
\prod_{l=1}^{\infty}(1-e^{2\pi i (l\rho)})^{\sum_s e^{-2\pi i sl/N}c^{(0,s)}(0)}.
\end{eqnarray}

\subsubsection*{Non-geometric orbifolds: $11A, 14A, 15A, 23A$ }

For completeness we note that we can extend the counting of hair modes to $g'$ orbifolds of $K3$ where $g'$ corresponds all the remaining  conjugacy classes of $M_{23}$. The CHL orbifolds also form a part of these, however the ones discussed in this section are non-geometric.
The hair modes in these cases can also be written as:
\begin{eqnarray}
Z_{{\rm hair}: g'}^{4d}&=& \prod_{l\ne 0}(1-e^{2\pi i (l\rho)})^{-\sum e^{-2\pi i sl/N}c^{(0,s)}(0)}
\end{eqnarray}
%\subsection{$g' \in [M_{23}]$}
To be explicit,  we list 
 list of values of $\sum_{s=0}^{N-1}e^{-2\pi i s l/N}c^{(0,s)}(-b^2)$ for different $N$ for $N=11,14,15,23$
 in table \ref{tablenong}. 
\begin{table}[H]
	\renewcommand{\arraystretch}{0.5}
	\begin{center}
		\vspace{0.5cm}
		\begin{tabular}{|c|c|c|c|}
			\hline
			& & &  \\
			$N$ & $l$ & $-b^2$ & $\sum_{s=0}^{N-1}e^{-2\pi i sl/N}c^{(0,s)}(-b^2)$\\ 
			& & &  \\
			\hline
			& & & \\
			$11$ & $11|l$ & 0 & $0$\\
			& & & \\
			& & $-1$ & 2\\
			\cline{2-4}
			& & & \\
			& $11\nmid l$ & 0 & 2\\
			& & & \\
%			& & $-1$ & 0\\
			\hline	\hline
			& & & \\
			$14$ & $14|l$ & 0 & $0$\\
			& & & \\
			& & $-1$ & 2\\
			\cline{2-4}
			& & & \\
			& $2| l,\; 7\nmid l$ & 0 & 2\\
			& & & \\
%			& & $-1$ & 0\\
			\cline{2-4}
			& & & \\
			& $7| l,\; 2\nmid l$ & 0 & 2 \\
			& & & \\
%			& & $-1$ & 0\\
			\cline{2-4}
			& & & \\
			& $2\nmid l, 7\nmid l$ & 0 & 1\\
			& & & \\
%			& & $-1$ & 0\\
			\hline			\hline
			& & & \\
			$15$ & $15|l$ & 0 & $0$\\
			& & & \\
			& & $-1$ & 2\\
			\cline{2-4}
			& & & \\
			& $3| l,\; 5\nmid l$ & 0 & 2\\
			& & & \\
%			& & $-1$ & 0\\
			\cline{2-4}
			& & & \\
			& $5| l,\; 3\nmid l$ & 0 & 2\\
			& & & \\
%			& & $-1$ & 0\\
			\cline{2-4}
			& & & \\
			& $3\nmid l, 5\nmid l$ & 0 & 1\\
			& & & \\
%			& & $-1$ & 0\\
			\hline	\hline	
					& & & \\
			$23$ & $23|l$ & 0 & $-2$\\
			& & & \\
			& & $-1$ & 2\\
			\cline{2-4}
			& & & \\
			& $23\nmid l$ & 0 & 1\\
			& & & \\
%			& & $-1$ & 0\\
			\hline		
		\end{tabular}
	\end{center}
	\vspace{-0.5cm}
	\caption{Values of  $\sum_{s=0}^{N-1}e^{-2\pi i s l/N}c^{(0,s)}(-b^2)$ for non-geometric orbifolds of $K3$ where $g'\in [M_{23}]$. $\sum_{s=0}^{N-1}e^{-2\pi i s l/N}c^{(0,s)}(-1)=0$ if $N\nmid l$ for any of these cases.} \label{tablenong}
	\renewcommand{\arraystretch}{0.5}
\end{table}
Using the results from table \ref{tablenong} we write:
\begin{eqnarray}\label{m23}
Z_{{\rm hair}: 11A}^{4d}&=&\prod_{l=1}^{\infty}(1-e^{2\pi i (11l\rho)})^{4}(1-e^{2\pi i (l\rho)})^{-2}(1-e^{2\pi i (11l\rho)})^{-2}\\
Z_{{\rm hair}:14A}^{4d}&=&\prod_{l=1}^{\infty}(1-e^{2\pi i (14l\rho)})^{4}(1-e^{2\pi i (14l\rho)})^{-1}\\ \nn
&&(1-e^{2\pi i (2l\rho)})^{-1}(1-e^{2\pi i (7l\rho)})^{-1}(1-e^{2\pi i (l\rho)})^{-1}\\ 
Z_{{\rm hair}:15A}^{4d}&=&\prod_{l=1}^{\infty}(1-e^{2\pi i (15l\rho)})^{4}(1-e^{2\pi i (15l\rho)})^{-1}\\ \nn
&&(1-e^{2\pi i (3l\rho)})^{-1}(1-e^{2\pi i (5l\rho)})^{-1}(1-e^{2\pi i (l\rho)})^{-1}\\
Z_{{\rm hair}: 23A}^{4d}&=&\prod_{l=1}^{\infty}(1-e^{2\pi i (23l\rho)})^{4}(1-e^{2\pi i (23l\rho)})^{-1}
(1-e^{2\pi i (l\rho)})^{-1}
\end{eqnarray}

The partition function of the horizon states in these models are given by 
the same expressions as in (\ref{horparti}) with $N$  
replaced by the order of the conjugacy class
and the coefficients $c^{(r, s)}$ read out out from the respective twisted elliptic genus. 
Let us conclude by writing the general formula for the horizon states as 
\begin{eqnarray}\nn
Z_{\rm{hor}: \; g'}^{4d}&=&-e^{-2\pi i(\rho+\sigma/N)} 
\prod_{r=0}^{N-1} \prod_{\begin{smallmatrix}k'\in \mathcal{Z}+r/N,l\in \mathcal{Z},\\ j\in \mathcal{Z}\\ k' >0,l\ge 0\end{smallmatrix}}(1-e^{2\pi i(k'\sigma+ l\rho+jv)})^{-\sum_s e^{-2\pi isl/N}c^{(r,s)}(4k'l-j^2)}\\ \label{horparti2}
&&\times \prod_{l=1}^{\infty}(1-e^{2\pi i (Nl\rho+v)})^{-2}\prod_{l=1}^{\infty}(1-e^{2\pi i (Nl\rho-v)})^{-2} (e^{\pi i v}-e^{-\pi i v})^{-2}
\end{eqnarray}

\subsection{Toroidal orbifolds}

In this section we construct the hair for ${\cal N}=4$ theories obtained by freely acting
 $\mathbb{Z}_2, \mathbb{Z}_3$ involutions on $T^6$ \cite{Sen:1995ff}. 
 Let us first briefly recall how these are constructed. 
 In the type IIB frame, they are obtained by 
$4$ of the co-ordinates together with a half shift along one of the $S^1$. 
The type IIA description of the theory is that of a freely acting orbifold with the action of $(-1)^{F_L}$ 
and a $1/2$ shift along one of the circles of $T^6$.\footnote{For details of these descriptions and 
the dyon configuration refer \cite{David:2006ji}.} 
A similar compactification of order 3 given by 
 a $2\pi/3$ rotation along one 2D plane of $T^4$ and a $-2\pi/3$ rotation 
along another plus an $1/3$ shift along one of the circles of
$T^2$ was also discussed in \cite{David:2006ji}. 
We call these models $\mathbb{Z}_2$ and $\mathbb{Z}_3$ toroidal orbifolds. 

One key property of these models to keep in mind which will be important is that
the breaking of the 32 supersymmetries  of type IIB to 16  is determined by the size of $S^1$. 
This was not the case for the orbifolds of $K3\times T^2$, where supersymmetry was 
broken by the $K3$.  For the toroidal models if the size of $S^1$ is infinite,  the theory 
effectively 
behaves as though the theory has 32 supersymmetries. 
We will use this fact to propose certain fermionic zero modes which were present 
for the CHL models will become singular at the horizon. 

The dyon partition function for the toroidal models  is given by \citep{David:2006ru}:
\begin{eqnarray}\label{siegform2}
\tilde{\Phi}_k(\rho,\sigma,v)&=&e^{2\pi i(\rho+  v)}\\ \nn
&&\prod_{b=0,1}\prod_{r=0}^{N-1}
\prod_{\begin{smallmatrix}k'\in \mathbb{Z}+
	\frac{r}{N},l\in \mathcal{Z},\\ j\in 2\mathbb{Z}+b\\ k',l\geq0, \; j<0\;  k'=l=0\end{smallmatrix}}
(1-e^{2\pi i(k'\sigma+l\rho+jv)})^{\sum_{s=0}^{N-1}e^{2\pi isl/N}c_b^{r,s}(4k'l-j^2)}.
\end{eqnarray}
%Note the difference in the factor on the first line in comparison with (\ref{siegform}). 
The 
coefficients $c^{(r,s)}$ are read out from the following twisted  elliptic genus for $\mathbb{Z}_2$ orbifold:
\begin{eqnarray}\label{2tortwist}
F^{(0, 0)} &=& 0 , \\ \nonumber
F^{(0, 1)} &=& \frac{8}{3} A(\tau, z) - \frac{4}{3} B(\tau, z) {\cal E}_2(\tau) 
, \\ \nonumber
F^{(1, 0)} &=& \frac{8}{3} A(\tau, z) + \frac{2}{3} B(\tau, z) {\cal E}_2(\frac{\tau}{2} ), 
\\ \nonumber
F^{(1, 1)} &=& \frac{8}{3}A(\tau, z) + \frac{2}{3} B(\tau, z) {\cal E}_2( \frac{\tau +1}{2} ).
\end{eqnarray} 
The corresponding Siegel 
form of weight $k =2$ can be written as 
\begin{equation}\label{phitwo}
\tilde \Phi_2 (\rho, \sigma, v) = \frac{ \tilde \Phi_{6}^2(\rho, \sigma, v) }{\tilde \Phi_{10}(\rho, \sigma, v) },
\end{equation}
where $\tilde\Phi_6$ is the weight $6$ Siegel modular form associated with the order 2  CHL orbifold. 
For the $\mathbb{Z}_3$ toroidal case the twisted elliptic genus is 
given by 
\begin{eqnarray} \label{3tortwist}
F^{(0, 0)} &=& 0 \\ \nonumber
F^{(0,s)}&=&A(\tau,z)-\frac{3}{4}B(\tau,z){\cal E}_3(\tau) \\ \nonumber
F^{(r,rk)}&=&A(\tau,z)+\frac{1}{4}B(\tau,z){\cal E}_3(\frac{\tau+k}{3}), \quad {r=1,2}.
\end{eqnarray} 
The Siegel modular form associated with the $\mathbb{Z}_3$ toroidal orbifold has weight $k=1$ and is given by 
\begin{equation}\label{phi1}
\tilde \Phi_1 (\rho, \sigma, v)  = \frac{\tilde\Phi_4^{3/2} (\rho, \sigma, v)}{\tilde\Phi_{10}^{1/2}(\rho, \sigma, v) },
\end{equation}
where $\tilde\Phi_4$ is the weight $4$ Siegel modular form associated with the order 3 CHL orbifold.

Let us construct the hair modes and horizon states  for these models. 

\subsubsection*{$T^6/\mathbb{Z}_2$ model}
\begin{itemize}
\item  Just as in the case of the CHL models, we have $4$ left moving fermions. 
This gives rise to 
\begin{equation}
Z_{{\rm hair}: T^6/\mathbb{Z}_2}^{4d: f }= \prod_{l =1}^\infty ( 1- e^{2\pi i (2l) \rho})^4.
\end{equation}
\item The deformations corresponding to the motion of the 
effective string in the  $3$ transverse  directions of $R^3\times \tilde S^1$ of the Taub-Nut space
together with the fluctuations of the anti-self dual forms can be determined easily by 
examining the  partition function of the fundamental string in this theory and removing the zero 
point energy. This partition function was determined in \cite{David:2006ji}, using this result we 
obtain \footnote{One can also obtain this by counting the number of invariant $2$-forms and the 
forms which pick up a phase as done in \cite{David:2006ud}. } 
\begin{equation}
Z_{{\rm hair}: T^6/\mathbb{Z}_2}^{4d\; b}  = \prod_{l =1}^\infty 
\left[ ( 1- e^{2\pi i (2l -1) \rho})^{8 }  ( 1- e^{4\pi i l  \rho})^{-8 } \right].
\end{equation}
\item Contribution of the zero modes:  
The quantum mechanics of the bosonic zero modes describing the motion of the D1-D5 system 
in the Taub-Nut 
result in the following partition function  \cite{David:2006yn}
\begin{equation}\label{torzero}
Z_{{\rm hair}: T^6/\mathbb{Z}_2}^{4d:\;{\rm zero modes}} =- e^{2\pi i v} ( 1- e^{2\pi i v} )^{-2} .
\end{equation}
For orbifolds of $K3$, this contribution from the bosonic zero modes 
 was cancelled by the zero modes of $4$ fermions from the
right moving sector carrying angular momentum $J= \pm \frac{1}{2}$ whose partition function is given by  
$-( e^{\pi i v} - e^{-\pi i v})^2$ \cite{Banerjee:2009uk}. 
However for the toroidal model,  we propose that these zero modes do not form part of the hair. 
They are either singular at the horizon or they are not localized outside the horizon. 
This is possible, 
 the fact that 
we are in a theory with $16$ supersymmetries  is tied
to the the radius of $S^1$.  Verification of this proposal would involve a detailed study of the 
zero mode wave functions which we leave for  the future. 
However we will  perform consistency checks of this proposal in section \ref{signind}.
by evaluating the index of the horizon states. 
\end{itemize}
Thus the hair modes of the $\mathbb{Z}_2$ toroidal model is given by 
\begin{equation}\label{4dtorhair1}
Z_{{\rm hair}: T^6/\mathbb{Z}_2}^{4d} =  -( e^{\pi i v} - e^{-\pi i v})^{-2}  \prod_{l =1}^\infty 
\left[ ( 1- e^{2\pi i (2l -1) \rho})^{8 }  ( 1- e^{4\pi i l  \rho})^{-4 } \right].
\end{equation}
The  partition function of the horizon states of this model are given by 
\begin{equation}
Z_{{\rm hor}: T^6/\mathbb{Z}_2}^{4d} = -\frac{1}{\tilde \Phi_2( \rho, \sigma, v)  
Z_{ {\rm hair}: T^6/\mathbb{Z}_2}^{4d}  }.
\end{equation}
where $\tilde\Phi_2(\rho, \sigma, v) $ is given in (\ref{phitwo}) or (\ref{siegform2}).

The toroidal model has another special feature, they admit Wilson lines along $T^4$ \cite{David:2006ru}, 
their partition function is given by 
\begin{eqnarray}\nn
Z_{{\rm Wilson}: T^4/\mathbb{Z}_2} =\prod_{l=1}^\infty \left[
(1- e^{2\pi i  (2l -1) \rho  + 2\pi i v })^2 (1- e^{2\pi i  (2l -1) \rho  - 2\pi i v })^2
(1- e^{2\pi i  (2l -1) \rho   })^{-4} \right]\\ \label{wilson1}.
\end{eqnarray}
It is possible that the Wilson lines might also be part of the hair modes.
In section \ref{signind} we will see that including the Wilson lines as hair modes instead of the
bosonic zero modes given in (\ref{torzero}) does not preserve the positivity of the index 
of the horizon states.

\subsubsection*{$T^6/\mathbb{Z}_3$ model}

Performing the same analysis as done for the $\mathbb{Z}_2$ orbfiold we obtain the 
following partition function for the hair modes.
\begin{equation}\label{4dtorhair2}
Z_{\rm{hair}: T^6/\mathbb{Z}_3}^{4d} = 
-(e^{\pi i v}-e^{-\pi i v})^{-2}\prod_{l=1}^\infty  \left[ \frac{ 
 (1-e^{2\pi i (3l-1)\rho})^{3} (1-e^{2\pi i (3l-2)\rho})^{3}}
 {(1-e^{2\pi i (3l)\rho})^{-2}} \right].
\end{equation}
The horizon states is given by 
\begin{equation}
Z_{\rm{hor}: T^6/\mathbb{Z}_3}^{4d}  =  -\frac{1}{\tilde \Phi_1(\rho, \sigma, v) 
Z_{\rm{hair}: T^6/\mathbb{Z}_3}^{4d}  } ,
\end{equation}
where $\tilde\Phi_1$ is given by (\ref{phi1}) or (\ref{siegform2}). 
For reference we also provide the partition function of the Wilson lines in this model
\begin{eqnarray}\label{wilson2}
&&Z_{{\rm Wilson}: T^4/\mathbb{Z}_3}  =  \nonumber \\ \nonumber
&&\prod_{l =1}^\infty 
\left[  \frac{ 
( 1-e^{2\pi i ( ( 3l -1)\rho  + v)} ) ( 1-e^{2\pi i ( ( 3l -2)\rho + v)} ) 
( 1-e^{2\pi i ( ( 3l -1)\rho - v) } )
( 1-e^{2\pi i ( ( 3l -2)\rho - v) } ) }
{( 1-e^{2\pi i ( ( 3l -1)\rho  )} )^2 ( 1-e^{2\pi i ( ( 3l -2)\rho )} )^2 }
\right] \nonumber \\
\end{eqnarray}

From the expression for the Wilson lines and the infinite product representation given 
for $\tilde \Phi_k $ given in  (\ref{siegform2}) we obtain the following useful expression for the 
partition function for the horizon modes  for both the toroidal orbifolds. 
\begin{eqnarray} \nonumber
&& Z_{\rm{hor}; \; T^6/\mathbb{Z}_N}^{4d} =e^{-2\pi i \rho}
 \prod_{r=0}^{N-1}
\prod_{\begin{smallmatrix}k\in \mathbb{Z}+r/N,l\in \mathbb{Z},\\ j\in \mathbb{Z}\\ k>0,l\ge 0\end{smallmatrix}}(1-e^{2\pi i(k\sigma+ l\rho+jv)})^{-\sum_s e^{-2\pi isl/N}c^{(r,s)}(4kl-j^2)}
\\ 
&&\times \prod_{l=1}^{\infty} \left[ 
(1-e^{2\pi i (Nl\rho+v)})^{-2} (1-e^{2\pi i (Nl\rho-v)})^{-2}   \right]
( e^{\pi i v} - e^{-\pi i v}) ^2 
 \times Z_{{\rm Wilson}: T^4/\mathbb{Z}_N}. \nonumber \\ \label{4dhortor}
\end{eqnarray}

\section{Horizon states for the BMPV black hole}\label{check}

We now examine the BMPV black hole in 5 dimensions, that is the transverse space
now does not have the Taub-Nut solution. The main reason for studying the problem 
in 5 dimensions is that the near horizon geometry of the BMPV black hole in 5 dimensions
is same as the of the $1/4$ BPS dyon in 4 dimensions. 
This implies that the partition function of the horizon states of these 2 systems should 
be identical. 
In this section we   construct the partition function of   the 
hair and the horizon states for the BMPV black hole in type IIB on  $K3\times S^1/g'$ 
as well 
as toroidal orbifolds of $T^5$.  Here $g'$ corresponds to all the conjugacy classes of $M_{23}$. 

\subsection{Partition function of BMPV black holes}

 The partition function for these black holes
in the canonical compactification $K3\times S^1$,  was constructed in 
\cite{Banerjee:2009uk}. 
The same analysis can  be extended to all the CHL models. 
The partition function receives contributions from the following sectors.
\begin{itemize}
\item The bound states of the D1-D5 system, this is given by the elliptic genus of the 
symmetric product of $K3/{g'}$. This contribution was evaluated in 
\cite{David:2006yn}. It is given by 
\begin{eqnarray}
Z^{5d}_{S^N K3/g'}= e^{-2\pi i\sigma/N} \prod_{r=0}^{N-1}
\prod_{\begin{smallmatrix}k\in \mathbb{Z}+r/N,l\in \mathbb{Z},\\ j\in \mathbb{Z}\\ k>0,l\ge 0\end{smallmatrix}}(1-e^{2\pi i(k\sigma+ l\rho+jv)})^{-\sum_s e^{-2\pi isl/N}c^{(r,s)}(4kl-j^2)}.
\nonumber \\
\end{eqnarray}
\item The centre of mass  motion of the D1-D5 system in flat space. The degrees of freedom 
consist of $4$ bosons and $4$ fermions. 2 pairs of  bosons carry the angular momentum $J=\pm 1$. 
\cite{Banerjee:2009uk}. 
\begin{equation}
Z^{5d}_{{\rm c.o.m}} = \prod_{l=1}^{\infty} \left[ (1-e^{2\pi i (Nl\rho+v)})^{-2}(1-e^{2\pi i (Nl\rho-v)})^{-2} 
(1-e^{2\pi i Nl\rho})^4 \right].
\end{equation}
Note that  the only difference from the canonical model is that the unit of 
momentum on $S^1$  is $N$ due to the $1/N$ shift. 
\item   $4$ right chiral zero modes which contribute as $(-1)^J e^{2\pi J}$ which contribute 
in pairs with $J = \pm\frac{1}{2}$ 
\begin{equation}
Z^{5d}_{{\rm zero modes}} = - ( e^{\pi i v} - e^{-\pi i v} )^2.
\end{equation}
\item A shift of $e^{-2\pi i \rho}$ to ensure to take into account of the difference in the 
electric charge measured at infinity and the horizon \cite{Banerjee:2009uk}. 
\end{itemize}
Combining all the sectors we obtain the following expression for the partition function 
for BMPV black hole for all orbifolds of $K3\times S^1$. 
\begin{eqnarray}\nn
Z^{5d}_{g'} &=& -e^{-2\pi i(\rho+\sigma/N)} \prod_{r=0}^{N-1}
\prod_{\begin{smallmatrix}k\in \mathcal{Z}+r/N,l\in \mathcal{Z},\\ j\in \mathcal{Z}\\ k>0,l\ge 0\end{smallmatrix}}(1-e^{2\pi i(k\sigma+ l\rho+jv)})^{-\sum_s e^{-2\pi isl/N}c^{(r,s)}(4kl-j^2)}\\ \nn
&& \times  \prod_{l=1}^{\infty} \left[ 
(1-e^{2\pi i (Nl\rho+v)})^{-2}(1-e^{2\pi i (Nl\rho-v)})^{-2} (e^{\pi i v}-e^{-\pi i v})^{2} (1-e^{2\pi i Nl\rho})^4.
\right]
\\ \label{d5d}
\end{eqnarray}
Here the coefficients $c^{(r, s)}$  have to be read out from the twisted elliptic genus of  $K3$ 
by $g'$ corresponding to the conjugacy classes of $M_{23}$.

Using the counting of states  for the dyon partition function done
in  \cite{David:2006ru} we can extend the analysis to the toroidal models. 
We present the analysis in some detail for the $T^6/\mathbb{Z}_2$ model
Here the contributions arise from the following:
\begin{itemize}
\item
The bound state of the D1-D5 system on the $T^4/\mathbb{Z}_2$ orbifold is given by 
\begin{eqnarray}
Z^{5d}_{S^N T^4/\mathbb{Z}_2 } &=&  \prod_{r=0}^{N-1}
\prod_{\begin{smallmatrix}k\in \mathbb{Z}+r/N,l\in \mathbb{Z},\\ j\in \mathbb{Z}\\ k>0,l\ge 0\end{smallmatrix}}(1-e^{2\pi i(k\sigma+ l\rho+jv)})^{-\sum_s e^{-2\pi isl/N}c^{(r,s)}(4kl-j^2)}\nonumber
\\
&& \qquad\qquad\quad  N=2 
\end{eqnarray}
Here the coefficients $c^{(r, s)}$ are read out from the expansion of the functions given in 
(\ref{2tortwist}). 
\item
The contribution of the Wilson lines on $T^4/\mathbb{Z}_2$ which is given by 
\begin{eqnarray}
Z^{5d}_{{\rm Wilson}: T^4/\mathbb{Z}_2} =\prod_{l=1}^\infty \left[
(1- e^{2\pi i  (2l -1) \rho  + 2\pi i v })^2 (1- e^{2\pi i  (2l -1) \rho  - 2\pi i v })^2
(1- e^{2\pi i  (2l -1) \rho   })^{-4} \right]\nonumber \\
\end{eqnarray}
\item The partition function corresponding to the centre of mass motion of the D1-D5 system
in the transverse space
\begin{equation}
Z^{5d}_{{\rm c.o.m}} = \prod_{l=1}^{\infty} \left[ (1-e^{2\pi i (Nl\rho+v)})^{-2}(1-e^{2\pi i (Nl\rho-v)})^{-2} 
(1-e^{2\pi i Nl\rho})^4 \right], \qquad N=2.
\end{equation}
\item The contribution of the zero modes
\begin{equation}
Z^{5d}_{{\rm zero modes}} = - ( e^{\pi i v} - e^{-\pi i v} )^2.
\end{equation}
\item  The shift in the electric charge  accounted for by the factor $e^{-2\pi i \rho}$. 
\end{itemize}
Combining all the contributions we obtain 
 \begin{eqnarray}\nn
&& Z^{5d}_{T^5/\mathbb{Z}_N}= -e^{-2\pi i \rho} \prod_{\begin{smallmatrix}k\in \mathcal{Z}+r/N,l\in \mathcal{Z},\\ j\in \mathcal{Z}\\ k>0,l\ge 0\end{smallmatrix}}(1-e^{2\pi i(k\sigma+ l\rho+jv)})^{-\sum_s e^{-2\pi isl/N}c^{(r,s)}(4kl-j^2)}\\ \nn
&& \times \prod_{l=1}^{\infty} \left[ 
(1-e^{2\pi i (Nl\rho+v)})^{-2}
(1-e^{2\pi i (Nl\rho-v)})^{-2} 
(1-e^{2\pi i Nl\rho})^4 \right]  (e^{\pi i v}-e^{-\pi i v})^{2}  \times Z_{{\rm Wilson}: T^4/\mathbb{Z}_N}
\\  & & \qquad \qquad \qquad N=2 .
\label{d5dtor}
\end{eqnarray}

The partition function of the BMPV black hole in the $T^5/\mathbb{Z}_3$  is obtained 
model is given by  the same expression as in (\ref{d5dtor}) except that the coefficients
$c^{(r, s)}$ must  be read out from the functions given in (\ref{3tortwist}) and 
$N\rightarrow 3$.

\subsection{Orbifolds of $K3\times S^1$ }

We now construct  the hair modes in $5$ dimensions  for the $K3\times S^1/g'$ 
where the quotient by $g'$ associated with any conjugacy classes of the Mathieu group
$M_{23}$. 
The analysis proceeds identical to that done in \cite{Jatkar:2009yd}, the only difference being that the 
unit of momentum on $S^1$ is $N$. 
Here we briefly state the contributions. 
\begin{itemize} 
\item The contribution of the 4 real left moving gravitino  deformations of the BMPV black hole \footnote{
The bosonic deformations were shown to be singular at the horizon in \cite{Jatkar:2009yd}. }.
\begin{eqnarray}
Z_{{\rm hair}: \; g'}^{5d; \; f}  = \prod_{l = 1}^\infty ( 1- e^{2\pi il N  \rho})^4.
\end{eqnarray}
\item The contribution of the $8$ real gravitino zero modes among the $12$ modes 
 due to broken supersymmetries which carry angular momentum $J=\pm \frac{1}{2}$. 
 \begin{eqnarray}
 Z_{{\rm hair}:\;  g' }^{5d; \; {\rm zero\;modes} }    = ( e^{\pi iv} - e^{- \pi i v} )^4.
 \end{eqnarray}
\end{itemize}
Combining these contributions we obtain 
\begin{eqnarray} \label{5dhair}
Z_{\rm{hair}: \; g'}^{5d} = ( e^{\pi iv} - e^{- \pi i v} )^4 \prod_{l = 1}^\infty ( 1- e^{2\pi il N  \rho})^4.
\end{eqnarray} 

The partition function for the horizon states is given by 
\begin{eqnarray}\label{5dhor}
Z_{{\rm hor}: \; g'}^{5d}  = \frac{Z^{5d}_{g'}}{ Z_{\rm{hair}: \; g'}^{5d} } .
\end{eqnarray}
Now comparing the horizon states of the $4d$ dyons from (\ref{horparti2}) and 
using (\ref{d5d}) and (\ref{5dhair}) in (\ref{5dhor}) we can easily conclude 
\begin{equation}
Z_{{\rm hor}: \; g'}^{5d} = Z_{{\rm hor}: \; g'}^{4d}.
\end{equation}

\subsection{Toroidal models}

For the toroidal models  the contributions of the hair are as follows. 
\begin{itemize}
\item The contribution of the $4$ left moving  gravitino modes  which result in 
\begin{eqnarray}
Z_{{\rm hair}: \; T^5/\mathbb{Z}_{N}}^{5d; \; f}  = \prod_{l = 1}^\infty ( 1- e^{2\pi il N  \rho})^4 , 
\qquad\qquad N =2, 3.
\end{eqnarray}
\item The contirbution of the zero modes. 
As we discussed earlier,  supersymmetry in these models is tied to the 
radius of $S^1$.  We propose that due to this, out of 
$8$  gravitino zero modes arising from broken supersymmetries 
which has angular momentum $J=\pm \frac{1}{2}$, the wave functions of $4$ of them 
either become singular at the horizon or they  not localized  outside the horizon. 
These $4$ modes should not be counted as hair modes. 
Therefore the contribution of the zero modes in these models are given by 
\begin{eqnarray}
Z_{{\rm hair}: \; T^5/\mathbb{Z}_{N}}^{5d; \; {\rm zero\; modes} } =- ( e^{\pi i v} - e^{- \pi i v  } )^2.
\end{eqnarray}
As we will see consistency checks for this proposal will be done in section (\ref{signind}). 
\end{itemize}
Combining these contributions we obtain
\begin{eqnarray}\label{5dhairtor}
Z_{\rm{hair}: \; T^5/\mathbb{Z}_N}^{5d} 
= -( e^{\pi iv} - e^{- \pi i v} )^2 \prod_{l = 1}^\infty ( 1- e^{2\pi il N  \rho})^4.
\end{eqnarray} 
The horizon partition function from the $5d$ perspective is given by 
\begin{equation}\label{5dhortor}
Z_{\rm{hor}: T^5/\mathbb{Z}_N}^{5d}  = \frac{ Z^{5d}_{T^5/\mathbb{Z}_N}}
{Z_{\rm{hair}: \; T^5/\mathbb{Z}_N}^{5d} }.
\end{equation}
Comparing  the $4d$ horizon partition function given in (\ref{4dhortor}) and using 
(\ref{d5dtor}) and (\ref{5dhairtor}) in (\ref{5dhortor}) we see that
\begin{equation}
Z_{\rm{hor}: T^5/\mathbb{Z}_N}^{5d}  = Z_{\rm{hor}: T^6/\mathbb{Z}_N}^{4d} .
\end{equation}

\section{The sign of the index for horizon states}
\label{signind}

In this section we will address the main goal of the paper. 
We observe that the index of  horizon states is always positive.

\subsection{Canonical example: $K3\times T^2$}

For the un-orbifolded model recall that  the hair in $4d$ is given by 
\be\label{hair2b}
Z_{ {\rm hair } : 1A}^{4d}=\prod_{l =1}^{\infty} \frac{1}{(1-e^{2\pi i l\rho})^{20}}.
\ee
The partition function of the horizon states is obtained by 
\begin{eqnarray} \label{horstate}
Z_{{\rm hor} :1A} &=& \frac{1}{\Phi_{10}( \rho, \sigma, v)  Z_{ {\rm hair } : 1A}^{4d}} 
=  \frac{ \prod_{l=1}^\infty ( 1- e^{2\pi i l \rho} ) }{ \Phi_{10} ( \rho, \sigma, v) }.
\end{eqnarray}
It was  observed in \cite{Sen:2010mz}  that the  index $-B_6$ or the 
Fourier coefficients of $1/\Phi_{10}$ extracted using the 
contour in (\ref{contour}) subject to the  kinematic restrictions
\begin{eqnarray} \label{keres}
Q.P \geq 0, \quad Q\cdot P \leq Q^2, \quad Q\cdot P \leq P^2,  \quad Q^2, P^2, ( Q^2 P^2 - 
(Q.P)^2 ) >0
\end{eqnarray}
were positive.  The contour together with the 
above kinematic constraints ensures that the index
counts single centred dyons. 
Further more 
 \citep{Bringmann:2012zr}  proved that  the index of all single centered dyons with 
 $P^2 = 2, 4$ is positive. 
 These works assumed that there existed a frame in which the fermionic zero modes 
 associated with 
 broken supersymmetries were the only hair. We have seen that the type IIB frame the
 hair degrees of freedom is given by (\ref{hair2b}). 
 Now naively it seems from the expression for the horizon states in (\ref{horstate}) 
 there are negative terms introduced due to the factor in the numerator and the observation
 of positivity seen in (\citep{Sen:2010mz}) and \citep{Bringmann:2012zr} might be violated once
 the hair in the type IIB frame is factored out. 
 However we will show by adapting the proof of  \citep{Bringmann:2012zr} that 
 single centred dyons with $P^2 =2$ do have positive index. 
 For other values of charges we evaluate the index numerically, our results 
 are presented in table \ref{k3}. 
 We observe that for single centered dyons the index is indeed positive.

%
%
%
%It was shown in \citep{Bringmann:2012zr} for $P^2=2,4$ that the Fourier coefficients of inverse Siegel form in the attractor chamber are positive. Removal of hair modes naively seem to remove this positivity. This however is not true and we show it here for $P^2=2$ following similar steps as in \citep{Bringmann:2012zr}. For higher order charges we shall not go into the proof but we state the observations to be positive \ref{k3}.
%Here we list the values of $P^2, Q^2, Q\cdot P$ and the degeneracies obtained from the Fourier coefficients of the $\frac{1}{\Phi_{10}}/Z_{ {\rm hair } : 1A}^{4d}$
%\be
%Z_{ {\rm hair } : 1A}^{4d}=\frac{1}{(1-e^{2\pi i l\rho})^{20}}.
%\ee
\begin{table}[H]
	\renewcommand{\arraystretch}{0.5}
	\begin{center}
		\vspace{0.5cm}
		\begin{tabular}{|c|c|c|c|c|c|}
			\hline
			& & & & & \\
			$(Q^2,\;P^2)\;\;\;$ \textbackslash $Q\cdot P$ & 0 & 1 & 2 & 3 &4\\ 
			& & & & & \\
			\hline
			& & & & & \\
			(2, 2) & 28944 & 13863 & 1608 & 327 &  0  \\
			(2, 4) & 761312  & 406296 & 72424  & 6936 & {$-$ {648}}\\
			(2, 6) & 12324920  & 6995541 & 1423152 & 96619 & { {$-13680$}}\\
			(2, 8) & 148800072  & 88006584 & 19366320 & 1152216 & {$-${164244}}\\
			(4, 2) & 272832 & 154236 & 28944 & 1836 & {$-$ {648}}\\
			(4, 4) & 12980224 & 8595680 & 2665376 & 406296 & 25760 \\
			(4, 6) & 333276712  & 235492308 & 85781820 & 16141380 & 1423152 \\
			(6, 6) & 6227822652   & 4771720755 & 2158667028 & 572268361 &  85781820 \\
			\hline			
		\end{tabular}
	\end{center}
	\vspace{-0.5cm}
	\caption{Index of horizon states for $K3\times T^2$,  note that  negative numbers have zero or negative values for $Q^2P^2 -(Q\cdot P)^2 $.}\label{k3}
	\renewcommand{\arraystretch}{0.5}
\end{table}

\subsubsection*{Proof of positivity at $P^2=2$}

We can do a Fourier expansion of $\frac{1}{\Phi_{10}(\tau,\sigma,z)}${\footnote{We use the variable $\tau$ instead of $\rho$ and $z$ in place of $v$ to keep consistency with previous work \citep{Chattopadhyaya:2018xvg}}} in terms of Jacobi forms.
\begin{equation}\label{fjdecomp}
\frac{1}{\Phi_{10}  ( q,  p, y ) }  =   \sum_{m = -1}^\infty \psi_m (\tau, z )  p^m, 
\qquad q = e^{2\pi i \tau}, p = e^{2 \pi i \sigma}, y = e^{2\pi i z} .
\end{equation}
$\psi_m(\tau, z)  \eta^{24}(\tau) $ is a weak Jacobi form of of weight 2 and index $m$. 
In \cite{Dabholkar:2012nd} it was shown that 
 $\psi_m(\tau, z) $ admits
the following decomposition 
\begin{equation}
\psi_m(\tau, z)  = \psi_m^{\rm P} (\tau, z)  + \psi_m^{\rm F} (\tau, z),
\end{equation}
where, $\psi_m^{\rm F} (\tau, z)$ has no poles in $z$.  The polar part 
is given by an Appell-Lerch sum:
\begin{eqnarray}
\psi_m^{\rm P} ( \tau, z) = \frac{p_{24}(m+1)}{\eta^{24}(\tau) } {\cal A}_{2, m }(\tau, z) , 
\\ \nonumber
{\cal A}_{2, m }(\tau, z) = \sum_{s\in \mathbb{Z}} 
\frac{q^{ms^2 +s}  y^{2ms +1}}{ (1 - q^s y )^2}  .
\end{eqnarray}
At $P^2=2$ we have $m=1$ and we can write
\be
\psi_1^{\rm F} (\tau, z)=-\frac{3}{\Delta}(E_4 B(\tau,z)+216{\cal H}(\tau,z)).
\ee
We need to show that $\psi_1^h=-\frac{3}{q\prod(1-q^n)^4}(E_4 B(\tau,z)+216{\cal H}(\tau,z))$ has the positivity property.
Here ${\cal H}$ is  the simplest  Jacobi mock modular form defined by  the Hurwitz-Kronecker class numbers
\begin{eqnarray}
{\cal H}(\tau, z)  = \sum_{n =0}^\infty   H( 4n - j^2) q^n y^l .
\end{eqnarray}
The coefficients $H(n)$ are defined by 
\begin{eqnarray}
H( n)  &=& 0 \qquad \hbox{for}  \; n <0, \\ 
 \sum_{ n\in\mathbb{Z} }  H(n) q^n &=& 
-\frac{1}{12} +\frac{1}{3} q^3  + \frac{1}{2} q^4 + q^7 + q^8  + q^{11}  + \cdots\\
{\cal H}(\tau,z) &=& \theta_3(2\tau, 2z)h_0(\tau)+\theta_2(2\tau, 2z)h_1(\tau).
\end{eqnarray}

We can write the weak Jacobi form $B(\tau,z)$ given in (\ref{ab}) as:
\be
B(\tau,z)=\frac{\theta_1^2(\tau,z)}{\eta^6}=\frac{1}{\eta^6}(\theta_2(2\tau)\theta_3(2\tau,2z)-\theta_3(2\tau)\theta_2(2\tau,2z))
\ee
where, $\theta_2(\tau,z)=\sum_{n\in\mathbb{Z}}q^{\frac{(n+1/2)^2}{2}}y^{n+1/2}$ and $\theta_3(\tau,z)=\sum_{n\in\mathbb{Z}}q^{n^2/2}y^n$ and $y=e^{2\pi i z}$.
So we see that even and odd powers of $y$ are separated in $\psi_1^F$ by the two theta functions. 
With this we can write $\psi_1^F$ and $\psi_1^h$ as follows:
\begin{eqnarray}\nn
\psi_1^F=\frac{3}{\Delta}\left(\theta_2(2\tau, 2z) \left(\frac{\theta_3(2\tau)}{\eta^6}E_4-216 h_1(\tau)\right)-\theta_3(2\tau, 2z)\left(\frac{\theta_2(2\tau)}{\eta^6}E_4+216 h_0(\tau)\right)\right)\\ \nn
\psi_1^h=\frac{3}{\Delta_4}\left(\theta_2(2\tau, 2z) \left(\frac{\theta_3(2\tau)}{\eta^6}E_4-216 h_1(\tau)\right)-\theta_3(2\tau, 2z)\left(\frac{\theta_2(2\tau)}{\eta^6}E_4+216 h_0(\tau)\right)\right),\\ \label{psi1h}
\end{eqnarray}
where $\Delta_4=q\prod_{n=1}^{\infty}(1-q^n)^4$.
We know the following results,
\begin{enumerate}
\item The Fourier coefficients in $h_0(\tau)$ and $h_1(\tau)$ are positive except for $q^0$ in $h_0(\tau)$ \cite{Bringmann:2012zr}.
\item All Fourier coefficients in the $q$ expansion of $\frac{\theta_2(2\tau)}{\eta^6} $ or $\frac{\theta_3(2\tau)}{\eta^6} $ are positive.
\item $E_4=1+240\sum_{n=1}^{\infty}\sigma_3(n) q^n$, where $\sigma_3(n)$ is given by, $\sum_{d, d|n} d^3$.
\end{enumerate}
Let us  observe the expression:
$\left(\frac{\theta_2(2\tau)}{\eta^6}E_4+216 h_0(\tau)\right)$. The only negative Fourier coefficient appears at $q^0$.
%For an even power of $z$ which corresponds to $Q\cdot P$ being even, the coefficients should be negative beyond $q^0$.
We can prove the folowing lemma:
\begin{lemma}
For a function $f(q)=-1+\sum_{n=1}^{\infty} a(n)q^n$ having all positive $a(n)$, the function $\frac{f(q)}{\prod_{n=1}^{\infty}(1-q^n)^k}$ has positive coefficients as long as $a(1)>k$ and $a(n+1)>k $ for all $n\in\mathbb{N}$.
\end{lemma}

\begin{proof}
We prove this for $\frac{1}{(1-q)^k}$ and then the rest can be similarly proved by using $q\rightarrow q^r$ and taking $f_{r+1}(q)=\frac{f_r(q)}{(1-q^{r+1})^k}$. For $f_2$ the coefficient of $q^1$ is, $a(1)-k>0$ and the coefficient of $q^N$  for $N>1$ is given by,
%$\frac{f(q)}{\prod_{n=1}^{\infty}(1-q)^k}$ is given by,
\[ -\binom{N+k-1}{N}+ \binom{N+k-2}{N-1}a(1)+ \binom{N+k-3}{N-2}a(2)\cdots > k.\]
\end{proof}
%{\Note:} For $N=1$ $a(1)-k>k$ for $k=4$.
We can write 
\be
\frac{1}{16}\left(\frac{\theta_2(2\tau)}{\eta^6}E_4+216 h_0(\tau)\right)=-1+\sum_{n=1}^{\infty}a(n)q^n.
\ee
Here $a(1)>15\sigma(1)>4$. Hence the removal of hair degrees of freedom ensures positivity of $-B_6$ for the sector $Q\cdot P={\rm even}$ when $Q^2\ge 0$.

In the series asociated with $\theta_2(2\tau,2z)$ in equation (\ref{psi1h}) 
the Fourier coefficient of $q^{n-1/4}$ is bounded from below by,
\[10\sigma_3(n)-9H(4n-1).\] Its positivity is ensured starting from $n=2$ using the following bounds:
\begin{enumerate}
\item $\sigma_3(n) \ge n^3$,
\item $H(n)<n$ \citep{Bringmann:2012zr}.
\end{enumerate}
 For $n=1$ the positivity still holds as $H(3)=1/3$. So the complete $q$ series expansion of $\left(\frac{\theta_3(2\tau)}{\eta^6}E_4-216 h_1(\tau)\right)$ contains no negative Fourier coefficient. This could also be seen from the Fourier expansion of $\left(\frac{\theta_3(2\tau)}{\eta^6}E_4-216 h_1(\tau)\right)$,
\be
\left(\frac{\theta_3(2\tau)}{\eta^6}E_4-216 h_1(\tau)\right)=q^{-1/4}(1+176q+\cdots).
\ee
This ensures the positivity of $-B_6$ for $Q\cdot P={\rm odd}$ and hence for $\psi_1^h$ as expected for $P^2=2$.

%\subsubsection*{Examples for higher order}

\subsection{Orbifolds of $K3\times T^2$}

For the $2A$ orbifold we extract the index of single centred dyons by using the contour 
in (\ref{contour}) together with the  following kinematic constraints on 
the charges \cite{Sen:2010mz}.
\begin{eqnarray}\label{tor2reg}
Q^2>0, \; P^2>0, \; Q.P\ge 0, \; P^2Q^2-(Q\cdot P)^2>0, \\ \nn
2Q^2\ge Q\cdot P, \; P^2 \ge Q\cdot P, \; P^2+2Q^2\ge 3 Q\cdot P.
\end{eqnarray}
The index of the horizon states for the $2A$ orbifold is given  in table \ref{chl2}.
\begin{table}[H]
	\renewcommand{\arraystretch}{0.5}
	\begin{center}
		\vspace{0.5cm}
		\begin{tabular}{|c|c|c|c|c|c|}
			\hline
			& & & & & \\
			$(Q^2,\;P^2)\;\;\;$ \textbackslash $Q\cdot P$ & 0 & 1 & 2 & 3 & 4\\ 
			& & & & & \\
			\hline
			& & & & & \\
			(1, 2) & 580 & 176 & $-2$ & 0 & 0\\
			(1, 4) & 5504 & 1856 & 32 & 0 & 0\\
			(1, 6) & 41476 & 16200 & 996 & 52 & 0\\
			(1, 10) &1293256 & 589200 & 63556 & 2752 &$-104$\\
			(2, 2) & 1312 & 576 & 48 & 0 & 0 \\
			(2, 4) & 16896 & 8640 & 1280 & 64 & 0 \\
			(3, 2) & 9708 & 4696 & 580 & 52 & 0\\
			\hline			
		\end{tabular}
	\end{center}
	\vspace{-0.5cm}
	\caption{Index of horizon states for the  $2A$ orbifold of $K3$ }\label{chl2}
	\renewcommand{\arraystretch}{0.5}
\end{table}

The kinematic constraints on the charges for the $3B$ orbifold  so that 
the dyons are single centered are given by 
\begin{eqnarray}
\{Q^2, \; P^2, P^2Q^2-(Q\cdot P)^2\}>0 \; Q.P\ge 0,
3Q^2\ge Q\cdot P, \; P^2 \ge Q\cdot P,\\ \nn
 2P^2+3Q^2\ge 5 Q\cdot P,\; P^2+6Q^2\ge 5 Q\cdot P, \; 2P^2+6Q^2\ge 7 Q\cdot P.
\end{eqnarray}
The index for the horizon states is then obtained using contour (\ref{contour}) and is 
listed in 
table \ref{chl3}.

\begin{table}[H]
	\renewcommand{\arraystretch}{0.5}
	\begin{center}
		\vspace{0.5cm}
		\begin{tabular}{|c|c|c|c|c|c|}
			\hline
			& & & & & \\
			$(Q^2,\;P^2)\;\;\;$ \textbackslash $Q\cdot P$ & 0 & 1 & 2 & 3 & 4\\ 
			& & & & & \\
			\hline
			& & & & & \\
			(2/3, 2) & 216 & 27 & 0 & 0 & 0 \\
			(2/3, 4) & 1548 & 342 & 0 & 0 & 0\\
			(2/3, 6) & 8532 & 2430 & 54 & 0 & 0\\
			(4/3, 2) & 540	& 216 & 0 & 0 & 0\\
			(4/3, 4) & 5820 & 2698 & 136 & 0 & 0\\
			(2, 2) & 1728 & 621 & 54 & 0 & 0\\
			(2, 6) & 204264 & 117837 & 23400 & 765 & 0\\
			(2, 8) & 1440288 & 896670 & 216540 & 13932 & $54$\\
			\hline
		\end{tabular}
	\end{center}
	\vspace{-0.5cm}
	\caption{Index of horizon states for the  $3A$ orbifold of $K3$}\label{chl3}
	\renewcommand{\arraystretch}{0.5}
\end{table}

For an orbifold of order $N>3$ there are infinite set of constraints for the charges to ensure that 
the index corresponds to single centered dyons
\cite{Sen:2010mz}. However we see as long as the norms of electric and magnetic charges are positive and $Q\cdot P\ge 0$ together with $Q^2P^2-(Q\cdot P)^2>0$, the index
 $-B_6$ remains positive for the orbifolds of $K3$ (see the tables \ref{table4b}-\ref{table23a}). These orbifolds maybe geometric like that of CHL or even non-geometric where $g'\in [M_{23}]$.

\begin{table}[H]
	\renewcommand{\arraystretch}{0.5}
	\begin{center}
		\vspace{0.5cm}
		\begin{tabular}{|c|c|c|c|c|c|}
			\hline
			& & & & & \\
			$(Q^2,\;P^2)\;\;\;$ \textbackslash $Q\cdot P$ & 0 & 1 & 2 & 3 & 4\\ 
			& & & & & \\
			\hline
			& & & & & \\
			(1/2, 2) & 64 & 8 & 0 & 0 & 0\\
			(1/2, 4) & 288 & 80 & 0 & 0 & 0\\
			(1/2, 6) & 1088 & 464 & 24 & 0 & 0\\
			(1, 2) & 96 & 48 & 0 & 0 & 0\\
			(1, 4) & 464 & 480 & 16 & 0& 0\\
			(3/2, 4) & 640 & 1680 & 160 & 0 & 0 \\
			(3/2, 6) & 3958 & 11448 & 2026 & 38 & 0\\
			(3/2, 22) & 232188670 & 421276388 & 228036842 & 43979890 & 2695862\\
			\hline			
		\end{tabular}
	\end{center}
	\vspace{-0.5cm}
	\caption{Index of horizon states for the $4B$  orbifold of $K3$} \label{table4b}
	\renewcommand{\arraystretch}{0.5}
\end{table}

\begin{table}[H]
	\renewcommand{\arraystretch}{0.5}
	\begin{center}
		\vspace{0.5cm}
		\begin{tabular}{|c|c|c|c|c|c|}
			\hline
			& & & & & \\
			$(Q^2,\;P^2)\;\;\;$ \textbackslash $Q\cdot P$ & 0 & 1 & 2 & 3 & 4\\ 
			& & & & & \\
			\hline
			& & & & & \\
			(2/5, 2) & 44 & 1 & 0 & 0 & 0  \\
			(2/5, 4) & 220 & 20  & 0 &0 & 0 \\
			(2/5, 6) & 880 & 125 & 0 & 0 & 0 \\
			(4/5, 2) & 88 & 16 & 0 & 0 & 0 \\
			(4/5, 4) & 560 & 160 & 0 & 0 & 0 \\
			(6/5, 6) & 8360 & 3755 & 310 & 0 & 0\\
			(6/5, 8) & 37394 & 18720 & 2202 & 16 & 0\\ 
			\hline
		\end{tabular}
	\end{center}
	\vspace{-0.5cm}
	\caption{Index of horizon states for the  $5A$ orbifold of $K3$} \label{table5a}
	\renewcommand{\arraystretch}{0.5}
\end{table}

\begin{table}[H]
	\renewcommand{\arraystretch}{0.5}
	\begin{center}
		\vspace{0.5cm}
		\begin{tabular}{|c|c|c|c|c|c|}
			\hline
			& & & & & \\
			$(Q^2,\;P^2)\;\;\;$ \textbackslash $Q\cdot P$ & 0 & 1 & 2 & 3 & 4\\ 
			& & & & & \\
			\hline
			& & & & & \\
			(1/3, 2) & 24 & 1 & 0 & 0 & 0\\
			(1/3, 4) & 92 & 12 & 0 & 0 & 0\\
			(1/3, 6) & 318 & 49 & 0 & 0 & 0\\
			(2/3, 2) & 44 & 10 & 0 & 0 & 0\\
			(2/3, 4) & 236 & 68 & 0 & 0 & 0\\
			(1, 4) & 564 & 216 & 8 & 0 & 0 \\
			(1, 6) & 2702 & 1201 & 100 & 0 & 0\\
			(1/3, 34) & 15836220 & 6614053 & 409414 & 1789 & $-14$\\
			\hline			
		\end{tabular}
	\end{center}
	\vspace{-0.5cm}
	\caption{Index of horizon states  $6A$ orbifold of $K3$}\label{table6a}
	\renewcommand{\arraystretch}{0.5}
\end{table}

\begin{table}[H]
	\renewcommand{\arraystretch}{0.5}
	\begin{center}
		\vspace{0.5cm}
		\begin{tabular}{|c|c|c|c|c|c|}
			\hline
			& & & & & \\
			$(Q^2,\;P^2)\;\;\;$ \textbackslash $Q\cdot P$ & 0 & 1 & 2 & 3 & 4\\ 
			& & & & & \\
			\hline
			& & & & & \\
			(2/7, 2) & 18 & 0 & 0 & 0 & 0\\
			(2/7, 4) & 72 & 3 & 0 & 0 & 0\\
			(2/7, 6) & 240 & 18 & 0 & 0 & 0\\
			(4/7, 2) & 30 & 3 & 0 & 0 & 0 \\
			(4/7, 4) & 150 & 31 & 0 & 0 & 0\\
			(6/7, 8) & 5580 & 2304 & 0 & 0 & 0\\
			(2/7, 40) & 46940778 & 18696804 & 1139238 & 4689 & $-18$\\
			\hline
		\end{tabular}
	\end{center}
	\vspace{-0.5cm}
	\caption{Index of horizon states for the  $7A$ orbifold of $K3$} \label{7a}
	\renewcommand{\arraystretch}{0.5}
\end{table}

\begin{table}[H]
	\renewcommand{\arraystretch}{0.5}
	\begin{center}
		\vspace{0.5cm}
		\begin{tabular}{|c|c|c|c|c|c|}
			\hline
			& & & & & \\
			$(Q^2,\;P^2)\;\;\;$ \textbackslash $Q\cdot P$ & 0 & 1 & 2 & 3 & 4\\ 
			& & & & & \\
			\hline
			& & & & & \\
			(1/4, 2) & 12 & 0 & 0 & 0 & 0\\
			(1/4, 4) & 40 & 2 &  0 & 0 & 0\\
			(1/4, 6) & 124 &  10 & 0 & 0 & 0\\
			(1/2, 2) & 20 & 2 & 0 & 0 & 0\\
			(1/2, 4) & 88 & 16 & 0 & 0 & 0\\
			(3/4, 4) & 176 & 52 & 0 & 0 & 0\\
			(3/4, 6) & 708 & 248 & 6 & 0 & 0 \\
			(1/4, 46) & 37469836 & 15088039 & 845410 & 2491 & $-10$\\
			\hline			
		\end{tabular}
	\end{center}
	\vspace{-0.5cm}
	\caption{Index of horizon states for the  $8A$ orbifold of $K3$} \label{table8a}
	\renewcommand{\arraystretch}{0.5}
\end{table}

It is interesting to see that the index for horizon states even in  non-geometric orbifolds of $K3$ retains positivity of the index in the domain $NQ^2\ge Q\cdot P, P^2 \ge Q\cdot P, Q^2 P^2-(Q\cdot P)^2>0.$
\begin{table}[H]
	\renewcommand{\arraystretch}{0.5}
	\begin{center}
		\vspace{0.5cm}
		\begin{tabular}{|c|c|c|c|c|c|}
			\hline
			& & & & & \\
			$(Q^2,\;P^2)\;\;\;$ \textbackslash $Q\cdot P$ & 0 & 1 & 2 & 3 & 4\\ 
			& & & & & \\
			\hline
			& & & & & \\
			(2/11, 2) & 6 & 0 & 0 & 0 & 0\\
			(2/11, 4) & 18 & 0 & 0 & 0 & 0\\
			(2/11, 6) & 50 & 1 & 0 & 0 & 0\\
			(4/11, 2) & 8 & 0 & 0 & 0 & 0\\
			(4/11, 4) & 32 & 4 & 0 & 0 & 0\\
			(6/11, 8) & 592 & 172 & 2 & 0 & 0\\
			(6/11, 10) & 1568 & 527 & 16 & 0 & 0\\
			\hline
		\end{tabular}
	\end{center}
	\vspace{-0.5cm}
	\caption{Index of horizon states for the 11A orbifold of $K3$}\label{table11a}
	\renewcommand{\arraystretch}{0.5}
\end{table}

\begin{table}[H]
	\renewcommand{\arraystretch}{0.5}
	\begin{center}
		\vspace{0.5cm}
		\begin{tabular}{|c|c|c|c|c|c|}
			\hline
			& & & & & \\
			$(Q^2,\;P^2)\;\;\;$ \textbackslash $Q\cdot P$ & 0 & 1 & 2 & 3 & 4\\ 
			& & & & & \\
			\hline
			& & & & & \\
			(1/7, 2) & 3 & 0 & 0 & 0 & 0\\
			(1/7, 4) & 7 & 0 & 0 & 0 & 0\\
			(1/7, 6) & 18 & 0 & 0 & 0 & 0\\
			(2/7, 2) & 4 & 0 & 0 & 0 & 0\\
			(2/7, 4) & 14 & 1 & 0 & 0 & 0\\
			(3/7, 8) & 163 & 45 & 0 & 0 & 0\\
			(3/7, 10) & 390 & 116 & 2 & 0 & 0\\
			(4/7, 10) & 774 & 329 & 14 & 0 & 0\\
			\hline
		\end{tabular}
	\end{center}
	\vspace{-0.5cm}
	\caption{Index of horizon states for the  $14A$  orbifold of $K3$}\label{table14a}
	\renewcommand{\arraystretch}{0.5}
\end{table}

\begin{table}[H]
	\renewcommand{\arraystretch}{0.5}
	\begin{center}
		\vspace{0.5cm}
		\begin{tabular}{|c|c|c|c|c|c|}
			\hline
			& & & & & \\
			$(Q^2,\;P^2)\;\;\;$ \textbackslash $Q\cdot P$ & 0 & 1 & 2 & 3 & 4\\ 
			& & & & & \\
			\hline
			& & & & & \\
			(2/15, 2) & 3 & 0 & 0 & 0 & 0\\
			(2/15, 4) & 6 & 0 & 0 & 0 & 0\\
			(2/15, 6) & 15 & 0 & 0 & 0 & 0\\
			(4/15, 2) & 3 & 1 & 0 & 0 & 0\\
			(4/15, 4) & 10 & 4 & 0 & 0 & 0\\
			(2/5, 8) & 125 & 31 & 0 & 0 & 0\\
			(2/5, 10) & 277 & 80 & 1 & 0 & 0\\
			(8/15, 10) & 527 & 227 & 9 & 0 & 0\\
			\hline
		\end{tabular}
	\end{center}
	\vspace{-0.5cm}
	\caption{Index of horizon states for the $15A$ orbifold of $K3$}\label{table15a}
	\renewcommand{\arraystretch}{0.5}
\end{table}

\begin{table}[H]
	\renewcommand{\arraystretch}{0.5}
	\begin{center}
		\vspace{0.5cm}
		\begin{tabular}{|c|c|c|c|c|c|}
			\hline
			& & & & & \\
			$(Q^2,\;P^2)\;\;\;$ \textbackslash $Q\cdot P$ & 0 & 1 & 2 & 3 & 4\\ 
			& & & & & \\
			\hline
			& & & & & \\
			(2/23, 2) & 1 & 0 & 0 & 0 & 0\\
			(2/23, 4) & 2 & 0 & 0 & 0 & 0\\
			(2/23, 6) & 5 & 0 & 0 & 0 & 0\\
			(4/23, 2) & 14 & 2 & 0 & 0 & 0\\
			(4/23, 4) & 28 & 4 & 0 & 0 & 0\\
			(6/23, 8) & 87 & 36 & 4 & 0 & 0\\
			(6/23, 10) & 144 & 57 & 6 & 0 & 0\\
			\hline
		\end{tabular}
	\end{center}
	\vspace{-0.5cm}
	\caption{ Index of horizon states for the $23A$  orbifold of $K3$}\label{table23a}
	\renewcommand{\arraystretch}{0.5}
\end{table}

\subsection{Toroidal orbifolds}

In \cite{Chattopadhyaya:2018xvg} we have seen that 
positivity of index for single centred dyons was violated for the toroidal models. 
For completeness we have reproduced some of the indices evaluated 
in  \cite{Chattopadhyaya:2018xvg}  in tables \ref{qp0}, \ref{qp1}, \ref{qp2}

\begin{table}[H] \footnotesize{
	\renewcommand{\arraystretch}{0.5}
	\begin{center}
		\vspace{0.5cm}
		\begin{tabular}{|c|c|c|c|c|}
			\hline
			 & & & & \\
			$Q^2\;\;\;$ \textbackslash \textbackslash $P^2$  & 2 & 4 & 6 & 8 \\ 
			& & & & \\
			\hline
			& & & &  \\
%			0 & 2 & 64 & 816 & 6912 & 45584 \\
			1 & {\bf -224} & {\bf -1248} & 1728 & 95104 \\
			2 & 1152 & 18240 & 233984 & 2432544 \\
			3 & {\bf -3392} &{\bf  -10320} & 542976 & 12103360 \\
			4 & -11520 & 200736 & 4575744 & 86712256 \\
			5 & {\bf -30336} & {\bf -55424} & 12914944 & 412163328 \\
			6 & 83968 & 1544832 & 61928448 & 2013023104 \\
			7 & {\bf -202560} & {\bf -179022} & 175358304 & 8292093664\\ 
			8 & 496512 & 9480000 & 638922240 & 32998944096 \\
			9 & {\bf -1118496} & {\bf -155232} & 1735394112 & 119618619520 \\
			10 & 2521600 & 49523328 & 5364983808 & 415768863360 \\
%			11&-275544 & {\bf  -5374656} & 2684560 & 13858160960 & 1359548367552 \\
%			12&581952 & 11389440 & 228872064 & 38347445760 & 4277873003392 \\
%			13&-1199688 & {\bf -23194176} & 24502656 & 94345755264 & 12874682948352 \\
%			14&2419584 & 46824960 & 959446272 & 240772494336 & 37480253184000 \\
%			15&-4783968 & {\bf -91770432} & 142728318 & 566613885216 & 105389524965472 \\
%			16&9288528 & 178117376 & 3712290336 & 1358448247296 & 288023853905856 \\
%			17&-17735256 & {\bf  -337839744} & 678230784 & 3072125756544 & 765208401512448 \\
%			18&33343344 & 634494592 & 13426540992 & 7004675317248 & 1983801614528672 \\
%			19&-61794600 & {\bf -1169806144 }& 2834120592 & 15289076372544 & 5022356020513856 \\
%			20&113002848 & 2136181248 & 45830851200 & 33439408301056 & 12447769083229056 \\
%			21&-204081024 & {\bf -3841753664} & 10783524096 & 7071452719680 & 30229570751178240 \\
%			22&364274496 & 6846494720 & 148756097664 & 149298142934016 & 72059338059045504 \\
%			23&-643092768 & {\bf -12044893632} & 38123432260 & 306899147706368 & 168747648892043328 \\
			\hline
		\end{tabular}
	\end{center}
	\vspace{0.5cm}
	\caption{The index $d(Q, P) $ for  the  $\mathbb{Z}_2$ toroidal orbifold
		 some  low lying values of $Q^2$, $P^2$ with $Q\cdot P=0$. }\label{qp0}
	\renewcommand{\arraystretch}{0.5}
	}
\end{table}

\begin{table}[H] \footnotesize{
	\renewcommand{\arraystretch}{0.5}
	\begin{center}
		\vspace{0.5cm}
		\begin{tabular}{|c|c|c|c|c|}
			\hline
			& & & & \\
			$Q^2\;\;\;$ \textbackslash \textbackslash $P^2$  & 2 & 4 & 6 & 8 \\ 
			& & & & \\
			\hline
			& & & &  \\
%			0&	0 & 0 & {-8} & { -128} & { -1160} \\
			1 & 96 & 1968 & 22528 & 190047 \\
			2 & {\bf -256} & 840 & 70912 & 1127672 \\
			3 & 1376 & 34656 & 728256 & 11046139 \\
			4 & {\bf -3840} & 16632 & 2497408 & 61486056 \\
			5 & 13152 & 343152 & 13144832 & 348876305 \\
			6 & {\bf -33536} & 171152 & 42058240 & 1603241304 \\
			7 & 92928 & 2476752 & 162898624 & 7016918625 \\
			8 & {\bf -220672} & 1265256 & 480911872 & 27503872048 \\
			9 & 540416 & 14545584 & 1556561664 & 102315259287 \\
			10 & {\bf -1204992} & 7558560 & 4271142656 & 354800345088 \\
%			11&	115154 & 2711616 & 73540080 & 12261114752 & 1175752005781 \\
%			12&	-251528 & {\bf -5741824} & 38736600 & 31586749312 & 3705255587616 \\
%			13&	534304 & 12144096 & 331284816 & 83106163712 & 11241057088056 \\
%			14&	-1107080 & {\bf -24613888} & 176485368 & 202830655232 & 32810366529704 \\
%			15&	2242936 & 49597408 & 1360242048 & 499048223424 & 92762004787995 \\
%			16&	-4452488 & {\bf -96865536} & 731764656 & 1162636791680 & 254219096542800 \\
%			17&	8675803 & 187681920 & 5172820416 & 2710918677760 & 678135519966520 \\
%			18&	-16618760 & {\bf -355014144} & 2806978216 & 6065899132672 & 1762706150153656 \\
%			19&	31335779 & 665705664 & 18435647328 & 13529566137472 & 4476930500026908 \\
%			20&	-58228616 & {\bf -1224694784} & 10082072832 & 29223048194432 & 11122701903357048 \\
%			21&	106740533 & 2233279616 & 62133135120 & 62776998234368 & 27083291897745248 \\
%			22&	-193201800 & {\bf -4009231104} & 34221009384 & 131432145572096 & 64699862426642976 \\
%			23&	345565877 & 7135993088 & 199430638848 & 273349419121472 & 151855990384385978 \\
			\hline
		\end{tabular}
	\end{center}
	\vspace{0.5cm}
	\caption{ The index $d(Q, P) $ for  the $\mathbb{Z}_2$ toroidal orbifold
		 some  low lying values of $Q^2$, $P^2$ with
	 $Q\cdot P=1$.  }\label{qp1}
	\renewcommand{\arraystretch}{0.5}
	}
\end{table}

\begin{table}[H] \footnotesize{
	\renewcommand{\arraystretch}{0.5}
	\begin{center}
		\vspace{0.5cm}
		\begin{tabular}{|c|c|c|c|c|}
			\hline
			& & & &  \\
			$Q^2\;\;\;$ \textbackslash \textbackslash $P^2$  & 2 & 4 & 6 & 8 \\ 
			& & & &  \\
			\hline
			& & & &   \\
%			0&	0 & 0 & 0 & 0 & 16 \\
			1&	 0 & {- 12} & {- 224} & { -1248} \\
			2&	64 & 2592 & 43264 & 491904 \\
			3&	 {\bf- 224} & 2432 & 191168 & 3805600 \\
			4&	 1152 & 43392 & 1440256 & 30853488 \\
			5&	 {\bf -3392} & 33720 & 5363680 & 171782688 \\
			6&	  11520 & 414336 & 24533248 & 893029504 \\
			7&	 {\bf -30336} & 302400 & 80281536 & 3963098880 \\
			8&	 83968 & 2926080 & 287831552 & 16432262672 \\
			9&	  {\bf -202560} & 2049968 & 851816352 & 62214237440 \\
			10&	 496512 & 16919712 & 2627695616 & 222752294016 \\
%			11&	-37508 & {\bf -1118496} & 11568000 & 7176834368 & 750069187008 \\
%			12&	86544 & 2521600 & 84554880 & 19942216704 & 2414262572768 \\
%			13&	-192800 & {\bf -5374656} & 56838432 & 51008186976 & 7425202332576 \\
%			14&	416512 & 11389440 & 377428608 & 131082715648 & 22009992439296 \\
%			15&	-875808 & {\bf -23194176} & 250745920 & 317429798336 & 62951326894880 \\
%			16&	1797648 & 46824960 & 1538196480 & 767174552576 & 174613994718000 \\
%			17&	-3610064 & {\bf -91770432} & 1013176056 & 1773519888864 & 470403008967552 \\
%			18&	7107328 & 178117376 & 5813224704 & 4077368575488 & 1234828601424128 \\
%			19&	-13741542 & {\bf -337839744} & 3805021440 & 9056713382272 & 3162966840870720 \\
%			20&	26130192 & 634494592 & 20608359552 & 19969539018240 & 7923569863533760 \\
%			21&-48930016 & {\bf -1169806144} & 13425820256 & 42839061178880 & 19436689033887616 \\
%			22&	90327040 & 2136181248 & 69137356032 & 91147913531648 & 46764751712533632 \\
%			23&	-164551050 & {\bf -3841753664} & 44883305472 & 189628816546240 & 110476832098945280 \\
			\hline
		\end{tabular}
	\end{center}
	\vspace{0.5cm}
	\caption{The index $d(Q, P) $ for  the  $\mathbb{Z}_2$ toroidal orbifold
		 some  low lying values of $Q^2$, $P^2$ with $Q\cdot P=2$. }\label{qp2}
	\renewcommand{\arraystretch}{0.5}
	}
\end{table}

\subsection*{Positivity of the horizon states for toroidal models}

The indices in tables \ref{qp0}, \ref{qp1}, \ref{qp2} 
were obtained 
 under the assumption that there exists a frame in which the fermionic
zero modes associated with broken supersymmetries are the only hair. 
In (\ref{4dtorhair1}) and (\ref{4dtorhair2}) we have proposed the partition function for the hair 
degrees of freedom in the type IIB frame for the $\mathbb{Z}_2, \mathbb{Z}_3$ 
toroidal orbifolds respectively. 
We evaluate the indices of horizon states in the following tables (\ref{tabletor1}-\ref{tabletor6}) and observe 
that they are all positive for single centered dyons.

%Now we observe the restoration of positivity after removing the hairs in the $T^6/\mathbb{Z}_2$ models.
%%\subsection{Torus orbifold $\mathbb{Z}_2$ without Wilson line and motion of $D1-D5$ in Taub-NUT space in the hair modes}
%Here we remove the transverse fermionic modes and the Kaluza Klein modes and right moving bosonic modes with $J=0$. The resultant factor that gets removed from $\frac{1}{\Phi_2}$  is given by
%\be
%-(e^{\pi i v}-e^{-\pi i v})^{-2} \prod_{n=1}^{\infty} \frac{(1-q^{2n-1})^8}{(1-q^{2n})^4}
%\ee

\begin{table}[H]
	\renewcommand{\arraystretch}{0.5}
	\begin{center}
		\vspace{0.5cm}
		\begin{tabular}{|c|c|c|c|c|}
			\hline
			& & & &  \\
			$Q^2\;$ {\large{\textbackslash\textbackslash}} $P^2$  & 2 & 4 & 6 & 8\\ 
			& & & &  \\
			\hline
			& & & &  \\
%	1 & 128 & 1584 & 13056 & 84288 \\
	1 & 832 & 14816 & 158848 & 1283902 \\
	2 & 3840 & 101008 & 1425920 & 14471264 \\
	3 & 14624 & 556176 & 10273024 & 129971582 \\
	4 & 48128 & 2588336 & 62037760 & 971443680 \\
	5 & 143424 & 10594400 & 325402624 & 6254176746 \\
	6 & 394112 & 39145344 & 1521266688 & 35582718576 \\
	7 & 1016080 & 133122060 & 6465235840 & 182481593350 \\
	8 & 2480512 & 422430736 & 25355844096 & 856661245280 \\
	9 & 5786240 & 1264061344 & 92844570752 & 3726638152610 \\
	10 & 12968576 &	3595680768 &	320340466176 &	15170555788976\\
			\hline
\end{tabular}
\end{center}
\vspace{-0.5cm}
\caption{Index of horizon states for the   $\mathbb{Z}_2$ orbifold of $T^6$ for $Q\cdot P=0$.}
\label{tabletor1}
\renewcommand{\arraystretch}{0.5}
\end{table}

\begin{table}[H]
	\renewcommand{\arraystretch}{0.5}
	\begin{center}
		\vspace{0.5cm}
		\begin{tabular}{|c|c|c|c|c|}
			\hline
			& & & &  \\
			$Q^2\;$ {\large{\textbackslash\textbackslash}} $P^2$  & 2 & 4 & 6 & 8\\ 
			& & & &  \\
			\hline
			& & & &  \\
%	1 & 64 & 784 & 6400 & 41024 \\
	1 & 480 & 9012 & 98784 & 811166 \\
	2 & 2496 & 69328 & 1001472 & 10329280 \\
	3 & 9888 & 403448 & 7664064 & 98689790 \\
	4 & 33664 & 1946480 & 48074496 & 766539920 \\
	5 & 102272 & 8155848 & 258619232 & 5063997322 \\
	6 & 286208 & 30667504 & 1231379200 & 29352001136 \\
	7 & 747456 & 105699406 & 5306269024 & 152656500694 \\
	8 & 1847040 & 339109664 & 21040306176 & 724593923536 \\
	9 & 4350816 & 1024054008 & 77737446688 & 3180401982114 \\
	10 & 9841408	&2935991504&	270248202752&	13043376086768\\
			\hline
\end{tabular}
\end{center}
\vspace{-0.5cm}
\caption{Index of horizon states for the   $\mathbb{Z}_2$ orbifold of $T^6$ for $Q\cdot P=1$.}
\label{tabletor2}
\renewcommand{\arraystretch}{0.5}
\end{table}

\begin{table}[H]
	\renewcommand{\arraystretch}{0.5}
	\begin{center}
		\vspace{0.5cm}
		\begin{tabular}{|c|c|c|c|c|}
			\hline
			& & & &  \\
			$Q^2\;$ {\large{\textbackslash\textbackslash}} $P^2$  & 2 & 4 & 6 & 8\\ 
			& & & &  \\
			\hline
			& & & &  \\
	1 & 96 & 1880 & 21056 & 178660 \\
	2 & 640 & 21312 & 329728 & 3577216 \\
	3 & 2992 & 151056 & 3115712 & 42306045 \\
	4 & 11008 & 813280 & 22062720 & 371908656 \\
	5 & 35840 & 3669600 & 128569280 & 2665839255 \\
	6 & 105472 & 14554120 & 647882496 & 16372365048 \\
	7 & 288192 & 52296704 & 2913889600 & 88924896642 \\
	8 & 738560 & 173535528 & 11950263808 & 436628175032 \\
	9 & 1798688 & 539123792 & 45385181120 & 1969579830259 \\
	10 & 4187008 & 1583791144 & 161466383616 & 8262793111120 \\
			\hline
		\end{tabular}
	\end{center}
	\vspace{-0.5cm}
	\caption{Index of horizon states for the  $\mathbb{Z}_2$ orbifold of $T^6$ for $Q\cdot P=2$.}
	\label{tabletor3}
	\renewcommand{\arraystretch}{0.5}
\end{table}

\begin{table}[H]
	\renewcommand{\arraystretch}{0.5}
	\begin{center}
		\vspace{0.5cm}
		\begin{tabular}{|c|c|c|c|c|}
			\hline
			& & & &  \\
			$Q^2\;$ {\large{\textbackslash \textbackslash}} $P^2$  & 2 & 4 & 6 & 8\\ 
			& & & &  \\
			\hline
			& & & &  \\
	1 & 0 & $-$12 & $-$224 & $-$1046 \\
	2 & 64 & 2480 & 40960 & 484752 \\
	3 & 320 & 26590 & 632544 & 9430780 \\
	4 & 1408 & 178096 & 5723136 & 106304080 \\
	5 & 5088 & 916872 & 38694432 & 887612004 \\
	6 & 16896 & 4001712 & 215960576 & 6052758272 \\
	7 & 50432 & 15481304 & 1047526432 & 35500683214 \\
	8 & 140352 & 54572672 & 4557481728 & 184959084864 \\
	9 & 365536 & 178371800 & 18160058144 & 874917932484 \\
	10 & 905600 & 547471520 & 67260039168 & 3817189761008 \\
			\hline
\end{tabular}
\end{center}
\vspace{-0.5cm}
\caption{Index of horizon states for the  $\mathbb{Z}_2$ orbifold of $T^6$ for $Q\cdot P=3$. Note that it is only when $Q^2P^2 - ( Q\cdot P)^2 <0$ we observe that the index is negative.}
\label{tabletor4}
\renewcommand{\arraystretch}{0.5}
\end{table}

\begin{table}[H]
	\renewcommand{\arraystretch}{0.5}
	\begin{center}
		\vspace{0.5cm}
		\begin{tabular}{|c|c|c|c|c|}
			\hline
			& & & &  \\
			$Q^2\;$ {\large{\textbackslash\textbackslash}} $P^2$  & 2 & 4 & 6 & 8\\ 
			& & & &  \\
			\hline
			& & & &  \\
	1 & 0 & 0 & 0 & 37 \\
	2 & 0 & $-8 $& $-256$ & 1232 \\
	3 & 16 & 1900 & 50880 & 868435 \\
	4 & 0 & 17928 & 757376 & 16261008 \\
	5& 96 & 114160 & 6613888 & 176919248 \\
	6 & 512 & 576016 & 43399680 & 1427632608 \\
	7 & 2416 & 2506512 & 236442496 & 9431113673 \\
	8 & 8320 & 9731384 & 1124958848 & 53751377384 \\
	9 & 26592 & 34532368 & 4818946176 & 272969682473 \\
	10 & 75904 & 113759408 & 18960610304 & 1262218427744 \\
			\hline
\end{tabular}
\end{center}
\vspace{-0.5cm}
\caption{Index of horizon states for the 
 $\mathbb{Z}_2$ orbifold of $T^6$ when  $Q\cdot P=4$. Note that only 
 when $Q^2P^2 - ( Q\cdot P)^2 <0$ we observe that the index is negative.}
 \label{tabletor5}
\renewcommand{\arraystretch}{0.5}
\end{table}

%{\color{red}{Note:}}

We now enumerate the consistency checks we have done for the proposal 
for the hair modes in the $T^6/\mathbb{Z}_2$ toroidal model given in  (\ref{4dtorhair1}).
\begin{enumerate}
\item If we do not include the 
zero modes  $- e^{2\pi i v} ( 1- e^{2\pi i v} )^{-2}$ 
as part of the hair partition function in 
 $T^6/\mathbb{Z}_2$, then we observe the violation of  positivity in index for 
 $P^2=6, Q^2=1, Q\cdot P=2$ and  $P^2=6, Q^2=2, Q\cdot P=3$. The indices 
  for these dyonic charges are $-224$ and $-256$ respectively. These charges are 
   within the kinematic domain 
 defined by (\ref{tor2reg}).

\item  If we include the  contribution of the Wilson lines
 given in (\ref{wilson1}) as part of the hair partition function
and remove the contribution of the zero modes
   $- e^{2\pi i v} ( 1- e^{2\pi i v} )^{-2}$,   we find violations in positivity of the index. 
   This can be observed at $P^2=6,\; Q^2=1,\; Q\cdot P=2$, $P^2=6,\; Q^2=2$, $Q\cdot P=3$, $P^2=4, \; Q^2=4, Q\cdot P=3$, the indices are  given by $-64, -64, -4$ respectively.
   \end{enumerate}
   
   These two observations show that we certainly need to include the contribution
   of the zero modes    $- e^{2\pi i v} ( 1- e^{2\pi i v} )^{-2}$  as part of the hair partition function
   which is consistent with our proposal.  It would be interesting to prove this by studying the 
   wave function of the gravitino zero modes in the toroidal models.

%\item If the right moving bosonic zero modes are not removed from the 4D hair positivity is not ensured in the region $P^2>0, Q^2>0$ and $Q\cdot P \ge 0$.
%\item We can try to remove the Wilson lines from the system and compute $\psi_m^F$. It is observed that there is a way of expanding the polar and finite components if the Wilson lines are removed. This can be given by,
%\be
%\frac{1}{Z_{wilson}\Phi^{torus}_k}=\sum_{m=0}^{\infty}\psi_m p^{m}.
%\ee
%However to remove the poles in this sector we need to subtract the Appell Lerch sums of order $m-1$ instead of $m$. This can ensure positivity at the level $P^2=2$ for $N=2$ toroidal orbifold however at the observed violation ($-64$) at $P^2=6,\; Q^2=2$ and $Q\cdot P=3$ the removal of poles will not be affected. This will remain unaffected by subtracting poles as ${\cal A}_{2,2}$ only starts to contribute from $Q\cdot P=5$.
%
%\item Wilson lines are not hair in this system. Their removal along with the whole proposed factor keeps the negative values for $-B_6$ when the area of the black hole is positive. However in Ashoke's region they still retain the positivity provided the right moving bosonic zero modes are removed.
%\item Removal of the factor $-(e^{\pi i v}-e^{-\pi i v})^{-2}$ ensures that the rest of the partition function is symmetric with $v\rightarrow -v$ and the region of positivity no more depends on $Q\cdot P\ge 0$.
%\item {\color{red}{Does it point to the fact that there is no way the positivity can be ensured in the attractor region? Please note that the attractor condition were all translated from $Q\cdot P \ge 0$.}}
A very similar analysis holds true for $T^6/\mathbb{Z}_3$. 
The index of horizon states obtained by considering the proposal given in (\ref{4dtorhair1})
for the hair partition function is positive as shown in the subsequent tables. 
We have also repeated the  consistency checks we mentioned earlier for the $\mathbb{Z}_2$ orbifold 
in this case with the same conclusions. 
%\be
%-(e^{\pi i v}-e^{-\pi i v})^{-2} \prod_{n=1}^{\infty} \frac{(1-q^{3n-1})^3(1-q^{3n-2})^3}{(1-q^{3n})^2}
%\ee

%\be
%-(e^{\pi i v}-e^{-\pi i v})^{-2} \prod_{l=1}^{\infty} (1-e^{2\pi i l\rho})^{-\sum_{s=0}^2 e^{-2\pi sl/3}c^{(0,s)}(0)}
%\ee

\begin{table}[H]
	\renewcommand{\arraystretch}{0.5}
	\begin{center}
		\vspace{0.5cm}
		\begin{tabular}{|c|c|c|c|c|c|}
			\hline
			& & & & & \\
			$(Q^2,\;P^2)\;\;\;$ \textbackslash $Q\cdot P$ & 0 & 1 & 2 & 3 & 4\\ 
			& & & & & \\
			\hline
			& & & & & \\
			(2/3, 2) & 162 & 90 & 9 & 0 & 0 \\
			(2/3, 4) & 1944 & 1134 & 162 & 0 & 0\\
			(2/3, 6) & 14598 & 8748 & 1149 & 0 & 0\\
			(4/3, 2) & 540 & 324 & 72 & 0 & 0\\
			(4/3, 4) & 8856 & 5724 & 1458 & 54 & 0\\
			(2, 2) & 1566 & 1008 & 243 & 18 & 0\\
			(2, 4) & 34344 & 23652 & 7290 & 810 & 0\\
			(2, 6) & 402972 & 286734  & 98613 & 13614 & 249 \\
			\hline
		\end{tabular}
	\end{center}
	\vspace{-0.5cm}
	\caption{Index of horizon states for the  $T^6/\mathbb{Z}_3$ orbifold}
	 \label{tabletor6}
	\renewcommand{\arraystretch}{0.5}
\end{table}

\section{Conclusions} \label{conclusions}

We have constructed the horizon partition function of the $1/4$ BPS dyonic black hole 
in ${\cal N}=4$ theories obtained by compactifying type IIB on orbifolds of $K3\times T^2$.
We then observed that the index of the horizon states of single centred black holes 
are all positive.  We adapted the proof of \cite{Bringmann:2012zr} and showed that
the  index of the horizon  partition function of
single centred dyons with $P^2=2$  remains positive. 

For the toroidal models we propose that the hair modes are given by (\ref{4dtorhair1}) and (\ref{4dtorhair1}). 
We showed the index of horizon states with this proposal is positive and performed consistency 
checks. 
As mentioned earlier it would be  interesting to study  the wave function 
of the zero modes of the gravitino in the toroidal models to check 
the proposal in (\ref{4dtorhair1}) and (\ref{4dtorhair2}). 
In \cite{Chattopadhyaya:2018xvg} it was noticed that 
that the  index of single centred dyons in these models were not positive when one 
assumed that the only hair modes are the Fermionic zero modes associated with 
broken supersymmetry generators.  Since hair modes are frame dependent, 
 the observations in this paper indicates that there is possibly no duality frame 
 for these models which contains only the Fermionic zero modes as the hair. 
 It will be interesting to verify this explicitly by an study similar to that 
 done in \cite{Chowdhury:2014yca,Chowdhury:2015gbk} for the 
 ${\cal N}=8$ theory. 

The observation that the index of horizon states in the canonical compactification on 
$K3\times T^2$ is positive is worth further study. It should be possible to extend the proof
of \cite{Bringmann:2012zr}  to higher values of $P^2$.

\vspace{.5cm}
{\bf Note Added:} As this work was nearing completion,  we became aware of the work done in 
\citep{Chakrabarti:2020ugm}.  The analysis of the hair modes 
 done for the CHL orbifolds of $K3$ in section \ref{sechorstate}  and 
\ref{check} overlaps with parts of \cite{Chakrabarti:2020ugm}. 

\acknowledgments{We thank Ashoke Sen for very 
useful  discussions  at several instances
over the course of this project which helped us to understand issues related to the positivity 
of the index. We also thank Jan Manschot for helpful discussions.
We thank 
Amitabh Virmani for discussions and informing us 
of the conclusions of \cite{Chakrabarti:2020ugm}. 
The work of A.C is funded by IRC Laureate Award 15175.}

%\bibliographystyle{JHEP}
%\bibliography{hair_ref} 

\providecommand{\href}[2]{#2}\begingroup\raggedright\endgroup

\end{document}